\documentclass[a4paper,UKenglish,cleveref, autoref, thm-restate]{lipics-v2021}

\usepackage{colortbl}
\definecolor{green}{rgb}{0.1,0.1,0.1}

\graphicspath{ {./Figs/} }

\newcommand{\poly}{\textrm{\textup{poly}}}
\newcommand{\supp}{\textrm{\textup{supp}}}
\usepackage{mathtools}

\crefname{ineq}{inequality}{inequalities}
\creflabelformat{ineq}{#2{\upshape(#1)}#3}

\usepackage[linesnumbered,lined,boxed,ruled,vlined]{algorithm2e}

\SetKwInOut{Parameters}{Parameters}
\SetKwComment{Comment}{$\triangleright$\ }{}
\SetKwInput{KwInput}{Input}
\SetKwInput{KwOutput}{Output}
\SetAlFnt{\small}
\SetAlCapFnt{\small}
\SetAlCapNameFnt{\small}
\SetAlCapHSkip{0pt}
\IncMargin{-\parindent}

\hideLIPIcs
\nolinenumbers

\newcommand{\omv}{\textnormal{\textsf{OMv}}}
\newcommand{\oumv}{\textnormal{\textsf{OuMv}}}
\newcommand{\ub}[1]{\ensuremath{O(N^{#1 - \epsilon})}}
\newcommand{\eb}[1]{\ensuremath{m^{#1 - \epsilon}}}
\newcommand{\ebn}[1]{\ensuremath{N^{#1 - \epsilon}}}

\DeclareMathOperator{\dist}{dist}

\newcommand{\pl}[1]{\ensuremath{\left\lfloor \frac{1}{\zeta(#1)} \cdot \frac{N}{d^{#1}} \right\rfloor}}

\usepackage{todonotes}

\title{Fine-Grained Complexity Lower Bounds for Families of Dynamic Graphs}
\author{Monika Henzinger}{Department of Computer Science, University of Vienna, Vienna, Austria}{monika.henzinger@univie.ac.at}{}{}
\author{Ami Paz}{LISN, CNRS \& Paris-Saclay University, France}{ami.paz@lisn.fr}{}{}
\author{A. R. Sricharan}{Department of Computer Science, UniVie Doctoral School Computer Science DoCS, University of Vienna, Vienna, Austria}{sricharan.arunapuram@univie.ac.at}{}{}

\funding{\textit{Monika Henzinger and A. R. Sricharan}: \flag{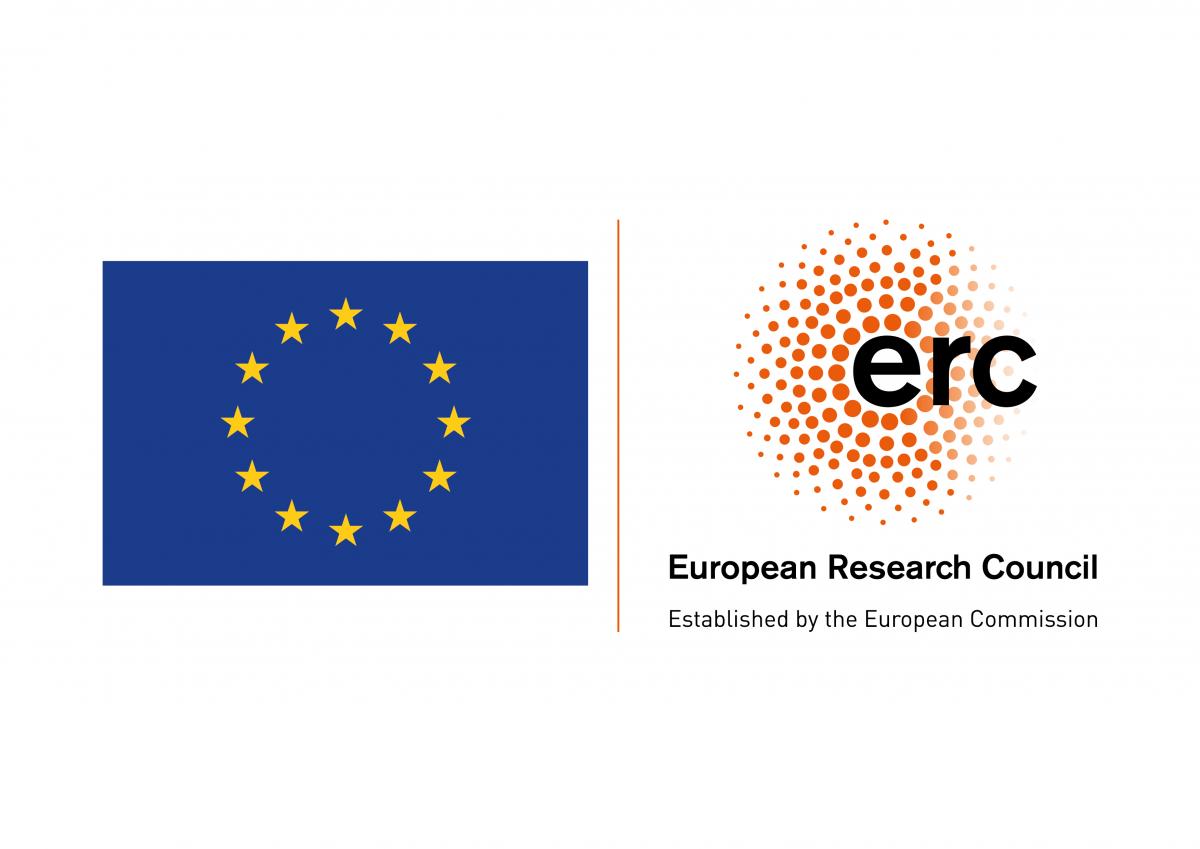}This project has received funding from the
European Research Council (ERC) under the European Union's Horizon 2020
research and innovation programme (Grant agreement No.\ 101019564
``The Design of Modern Fully Dynamic Data Structures (MoDynStruct)''
and from the
Austrian Science Fund (FWF) project ``Fast Algorithms for a Reactive Network
Layer (ReactNet)'', P~33775-N, with additional funding from the \textit{netidee SCIENCE
Stiftung}, 2020--2024}

\EventEditors{Shiri Chechik, Gonzalo Navarro, Eva Rotenberg, and Grzegorz Herman}
\EventNoEds{4}
\EventLongTitle{30th Annual European Symposium on Algorithms (ESA 2022)}
\EventShortTitle{ESA 2022}
\EventAcronym{ESA}
\EventYear{2022}
\EventDate{September 5--9, 2022}
\EventLocation{Berlin/Potsdam, Germany}
\EventLogo{}
\SeriesVolume{244}
\ArticleNo{34}
\Copyright{Monika Henzinger, Ami Paz, and A. R. Sricharan}

\date{}
\authorrunning{M. Henzinger, A. Paz, and A. R. Sricharan}

\begin{document}

\maketitle

\begin{abstract}
A dynamic graph algorithm is a data structure that answers queries about a property of the current graph while supporting graph modifications such as edge insertions and deletions. Prior work has shown strong conditional lower bounds for general dynamic graphs, yet graph families that arise in practice often exhibit structural properties that the existing lower bound constructions do not possess.
We study three specific graph families that are ubiquitous, namely constant-degree graphs, power-law graphs, and expander graphs, and give the first conditional lower bounds for them. Our results show that even when restricting our attention to one of these graph classes, any algorithm for fundamental graph problems such as distance computation or approximation or maximum matching, cannot simultaneously achieve a sub-polynomial update time and query time.
For example, we show that the same lower bounds as for general graphs hold for \emph{maximum matching} and \emph{($s,t$)-distance} in constant-degree graphs, power-law graphs or expanders. Namely, in an $m$-edge graph, there exists no dynamic algorithms with both $ O(m^{1/2 - \epsilon})$ update time and $ O(m^{1 -\epsilon})$ query time, for any small $\epsilon > 0$. Note that for ($s,t$)-distance the trivial dynamic algorithm achieves an almost matching upper bound of constant update time and $O(m)$ query time.
We prove similar bounds for the other graph families and for other fundamental problems such as densest subgraph detection and perfect matching.
\end{abstract}
\newpage
\section{Introduction}
\label{sec:Introduction}
A dynamic graph algorithm is a data structure that stores a graph and supports update operations, usually consisting of edge insertions and deletions, as well as query operations that ask about a specific property of the graph.
The introduction of strong conditional lower bounds based on widely-believed complexity assumptions~\cite{AbboudD16,HenzingerKNS15} has had a fundamental influence on the field, pushing the design of new algorithms towards more specialized algorithms such as partially-dynamic or even offline-dynamic algorithms or towards approximate solutions. However, graphs arising in real-world applications often differ significantly from the very specifically crafted  graphs for which  the lower bound results are shown.
Frequently, real-world graphs have some special structure, such as having a power-law degree distribution, a constant degree,  or being planar. Expanders, on the other side, have recently
been used to design dynamic algorithms for \emph{general} graphs.
This naturally leads to the question of determining the complexity of dynamic graph algorithms for these graph classes, and this is exactly the question investigated in this paper.

While the complexity of dynamic graph algorithms for planar graphs has already been studied quite extensively~\cite{DBLP:conf/esa/Subramanian93,DBLP:journals/jcss/HenzingerKRS97,klein1998fully,DBLP:conf/stoc/ItalianoNSW11,abraham12fully,Abraham2016,AbboudD16,DBLP:conf/stoc/ItalianoKLS17,CharalampopoulosK20}, the question is still widely open for other families of graphs, including power-law graphs, constant-degree graphs, and expanders.
Certain problems become easier for these graph classes:
As an $N$-node\footnote{To avoid confusion with the parameter $n$ and matrix $M$ used in the online-matrix-vector multiplication conjecture, we use $N$ to denote the number of vertices and $m$ the number of edges in the dynamic graphs.} constant-degree graph has $O(N)$ edges, computing all-pairs shortest paths (APSP) takes only time $\tilde O(N^2)$, while the popular APSP conjecture postulates that for general graphs, there exists a small constant $c > 0$ such that any algorithm in the word RAM model with $O(\log N)$-bit words requires $N^{3-o(1)}$ expected time
to compute APSP.
Moreover, some problems become trivial in these graph classes, e.g., computing shortest paths with logarithmic additive error on expander graphs is trivial, due to their low diameter.

In this paper we will concentrate on graph problems that have real-world applications such as shortest-paths (which has applications in online navigation), matching (which has applications in reconfigurable datacenters), and densest subgraphs (which has applications in network analysis),
yet we believe that our general approach can be applied to further graph problems.
For these three problems,
the known conditional lower bounds construct graphs that are far from being in the classes we consider: They have maximum degree $\Omega(N)$ and
small cuts, and their degree distribution is unknown as it depends on the instance that is postulated to be hard.

\subparagraph{Constant-degree graphs.} Various dynamic graph problems that admit strong lower bounds in general graphs have very efficient algorithms on constant-degree graphs.
Let $\Delta$ be the maximum degree in the graph.
For \emph{local} problems, where the solution at a node $v$ can be computed by simply analyzing information stored at the  neighbors of $v$ such as maintaining a maximal matching,
a  maximal independent set, or a $(\Delta+1)$-vertex coloring,
there exist simple dynamic algorithms with $O(\Delta)$ update time and constant query time.
Additionally, for various problems that count or detect certain fixed subgraphs with $c$ nodes (such as  triangle counting for $c=3$)
there exists  dynamic algorithms with $O(\Delta^{c-1})$ update time and constant  query time, even though they have polynomial conditional lower bounds in general graphs (see Table~\ref{tab:easy_on_const_degree}).
These efficient algorithms for local problems rule out the possibility of any non-trivial lower bound in the constant-degree setting.

Furthermore, even for the non-local problem of maintaining a maximum matching
Gupta and Peng~\cite{GuptaPeng13} designed a $(1+\delta)$-approximation algorithm for any small $\delta > 0$
that runs in $O\left( \min \left\{ \frac{m}{|M(t)|}, |M(t)| \right\} \delta^{-2} \right) $ amortized time per update, where $M(t)$ denotes the maximum cardinality matching after the  $t$-th update operation.
As in a graph with maximum degree $\Delta$ it holds that $|M(t)| \ge m/(2\Delta)$, their algorithm achieves
an amortized update time of $O(\Delta\, \delta^{-2})$, which is
$O(\delta^{-2})$ in constant-degree graphs.
This raises the question of how efficiently other non-local dynamic graph problems such exact maximum matching, shortest paths, and densest subgraph  can be solved
in dynamic constant-degree graphs and whether it is possible to show  (conditional) lower bounds for them.

\begin{table}[t]
    \small
    \renewcommand{\arraystretch}{1.2}
    \centering
    \begin{tabular}{|c|c|c|c|c|c|c|c|}
    \hline
    \multirow{3}{*}{\textbf{Problems}} & \multicolumn{5}{c|}{Lower bounds} & \multicolumn{2}{c|}{Upper bounds} \\ \cline{2-8}
    & \multicolumn{2}{c|}{General graphs~\cite{HenzingerKNS15}} & \multicolumn{3}{c|}{Erd{\H o}s-R\'enyi avg-case~\cite{henzinger2022complexity}} & \multicolumn{2}{c|}{constant $\Delta$ (trivial)} \\ \cline{2-8}
    & $u(m,N)$ & $q(m,N)$ & $p(m,N)$ & $u(m,N)$ & $q(m,N)$ & $u(m,N)$ & $q(m,N)$ \\ \hline
    Triangle & \multirow{6}{*}{\eb{1/2}} & \multirow{6}{*}{\eb{1}} & \multirow{4}{*}{\ebn{3}} & \multirow{4}{*}{\eb{1/2}} & \multirow{4}{*}{\eb{1}} & \multirow{2}{*}{$\Delta$} & \multirow{2}{*}{$1$} \\[-1pt]
    counting & & & & & & & \\ \cline{1-1} \cline{7-8}
    $C_4$ subgraph & & & & & & \multirow{2}{*}{$\Delta^3$} & \multirow{2}{*}{$1$}\\[-1pt]
    counting & & & & & & & \\ \cline{1-1} \cline{4-8}
    $5$-length & & & \multirow{2}{*}{\ebn{2}} & \multirow{2}{*}{\ebn{\omega-2}} & \multirow{2}{*}{1} & \multirow{2}{*}{$\Delta^{4}$} & \multirow{2}{*}{$1$} \\[-1pt]
    $(s,t)$-path count & & & & & & & \\ \hline

    \end{tabular}
    \caption{Counting problems which admit polynomial conditional lower bounds on general graphs (amortized) and on Erd{\H o}s-R\'enyi graphs (average case), but have algorithms with constant update and query times in constant-degree graphs. For the lower bounds above, there is no dynamic algorithm with pre-processing time $p(m,N)$, update time $u(m,N)$, and query time $q(m,N)$ unless the \omv{} conjecture is false. When $p(m,N)$ is unspecified, $\poly(N)$ pre-processing time is allowed.}
    \label{tab:easy_on_const_degree}
\end{table}

\subparagraph{Expanders.}
Expander decompositions are increasingly becoming a central tool for designing  dynamic graph algorithms with improved running time bounds for various graph problems such as connectivity, minimum spanning tree, shortest paths, conductance, edge-connectivity, maximum flows, and minimum cuts~\cite{DBLP:journals/corr/abs-2004-07650,DBLP:conf/soda/GoranciRST21,DBLP:conf/soda/ChuzhoyS21,DBLP:conf/stoc/Chuzhoy21}. One of the central subproblems that these algorithms have to handle is to solve a graph problem on a dynamically changing expander.
To understand the limitations of this approach it is crucial to understand which problems can be solved efficiently on expanders, and which cannot.
We present novel lower bounds for dynamic problems on expanders, more specifically on constant-degree expanders.

Note that these results also have an interesting connection to the average-case hardness of dynamic graph algorithms.
Recently, lower bounds on the average-case hardness were shown for various subgraph counting problems  in dynamic {Erd{\H{o}}s-R{\'{e}}nyi} graphs (see Table~\ref{tab:easy_on_const_degree} for some of them)~\cite{henzinger2022complexity}.
As random graphs are usually expanders,
giving lower bounds for a problem on dynamic expanders gives an indication that this problem might also be hard in the average case and can motivate further work in this direction.

\subparagraph{Power-law graphs.} Graphs are called \emph{power-law graphs} if the fraction of nodes with degree $d$ is proportional to $1/d^c$ for some constant $c > 0$.
Static and dynamic power-law graphs arise surprisingly often in real-world networks, such as the Internet, the world-wide web, and citation graphs, as well as in physics, linguistics, and economics.
Even though the existence of large dynamic power-law graphs was already pointed out in~2004~\cite{graham2004large}, no efficient dynamic algorithms have been developed specifically for this class of graphs. This leads to the question of whether sub-polynomial time dynamic algorithms are even possible for power-law graphs or not.
In fact,  dynamic power-law graphs were not only never studied, they were not even defined---removing even a single edge from a power-law graph changes the degree distribution and thus violates the power-law distribution.
Hence, we first present several definitions of \emph{dynamic} power-law graphs, where some slackness in the degree-distribution is allowed. Then we prove lower bounds that hold for all of these dynamic power-law graph definitions.

\subsection{Our Results}
\label{subsec:Our_Results}
Throughout the paper
we use the standard assumption that queries output one value, such as the size, length or weight of the solution.
Note that this makes it only more challenging to prove lower bounds.
All our results are conditioned on the popular \omv{} conjecture~\cite{HenzingerKNS15}, defined in Section~\ref{sec:Preliminaries}, but to  simplify the terminology we usually drop the word ``conditional''.
Our results are also summarized in
\cref{tab:our_results}.
\begin{enumerate}
\item \emph{Main results.}
We study the hardness of dynamic algorithms for (i) constant-degree graphs, (ii) expanders, and (iii) power-law graphs, for the following graph problems:  Determining (a) the size of a maximum matching, (b)  the length of the $(s,t)$ shortest-path (i.e. \emph{$(s,t)$-distance}), and (c)  the density of the densest subgraph. Specifically, we show the following tradeoff between the update time $u$ and the query time $q$ in an $m$-edge graph for maximum matching and $(s,t)$-distance:
There is no dynamic algorithm which achieves both $u = O(m^{1/2 - \epsilon})$ and $q = O(m^{1 -\epsilon})$ for any small $\epsilon > 0$.
Note that these bounds match the bounds given for general graphs in~\cite{HenzingerKNS15} and that the lower bound for $(s,t)$-distance is almost tight as the simple algorithm that only records the edge change at update time and computes the solution from scratch at query time achieves
$u=O(1)$ and $q = O(m)$.
For densest subgraph we show that
there is no dynamic algorithm which achieves both $u = O(N^{1/4 - \epsilon})$ and $q = O(N^{1/2 -\epsilon})$ for any small $\epsilon > 0$, which is weaker than the lower bound on general graphs (of $u = \ub{1/2}$ and $q=\ub{1}$).

The only relevant prior work are conditional lower bounds for planar graphs~\cite{AbboudD16}, which have constant degree: In unweighted graphs they show for all-pairs-shortest paths a weaker tradeoff between update time $u$ and query time $q$ than we do, namely they prove
$\max(u^2 \cdot q, u \cdot q^2) = \Omega(m^{1 - o(1)})$. In \emph{weighted} graphs they show for ($s,t$)-distance a tradeoff of
$\max(u, q) = \Omega(m^{1/2 - o(1)})$.
Note that our result is stronger as it shows that in \emph{unweighted} graphs no algorithm with $u = \Omega(m^{49/100})$
and $q = \Omega(m^{99/100})$
is possible.

\item \emph{Degree--lower bound trade-off.}
While the constant-degree lower bounds are equal to the lower bounds for general graphs in terms of $m$, they are naturally quadratically lower in terms of the number of nodes $N$.
To understand the behaviour of the bounds also with respect to $N$,
we extend our constant-degree lower bounds for maximum matching, perfect matching, and $(s,t)$-distance to graphs with maximum degree $O(N^t)$, for any $0 \leq t \le 1$.
We show the following result:
There is no dynamic algorithm which achieves both $u = O(N^{(1+t)/2 - \epsilon})$ and $q = O(N^{1 + t -\epsilon})$ in a graph with maximum
degree $O(N^t)$ for any $\epsilon > 0$.
These results hold even in bipartite graphs.
Note that for $t=1$ these results match exactly the bounds for general graphs in~\cite{HenzingerKNS15},
and for $t=0$, they match the aforementioned results for constant-degree graphs.

\item \emph{Approximation results.}
In constant-degree graphs we extend the lower bound to the problem of approximating the  $(s,t)$-distance within a factor of $3-\delta$, for any small constant~$\delta$. This naturally extends the $(3-\delta)$-approximation lower bounds on general graphs to the constant-degree case.
In planar graphs, the $(s, t)$-distance lower bound holds only for exact answers.

Note that a similar extension to approximation algorithms is not possible for maximum cardinality matching and for densest subgraph: (a) For maximum matching, for any small $\delta > 0$ the above-mentioned $(1+\delta)$-approximation algorithm~\cite{GuptaPeng13} achieves an
amortized update time of $O(\delta^{-2})$, which is constant for
a constant $\delta$, thereby precluding any non-trivial lower bounds for approximate maximum matching in the constant-degree setting. Stated differently, \emph{our work shows an interesting dichotomy for dynamic matching matching in constant-degree graphs}: For the exact setting there is no dynamic algorithm which achieves both $u = O(m^{1/2 - \epsilon})$ and $q = O(m^{1 -\epsilon})$ for any small $\epsilon > 0$, while a $(1+\delta)$-approximation can be achieved in constant time, for any small $\delta >0$.
(b) \emph{The same dichotomy arises for densest subgraph}: For any small $\delta>0$ there exists a $(1-\delta)$-approximation algorithm with polylogarithmic time per operation~\cite{sawlani20near}, while we show a polynomial lower bound for the exact value.

\item \emph{Partially dynamic algorithms.}
We extend the constant-degree reductions for maximum matching and $(s,t)$-distance
also to the insertions-only and the deletions-only setting, achieving the same lower bound as in the fully dynamic setting.

\item\emph{Perfect matching.}
A special case of maximum cardinality matching is determining
whether a perfect matching exists in a bipartite graph.
For constant-degree graphs and expander graphs we show
the following lower bound: There is no dynamic algorithm which achieves both $u = O(m^{1/2 - \epsilon})$ and $q = O(m^{1 -\epsilon})$ for any small $\epsilon > 0$.
This can also be extended to the varying-degree setting.

\end{enumerate}

To summarize, our paper opens up the research field of dynamic graph algorithms for more specific, practical graph classes, in contrast with previous work that concentrated on general or planar graphs. We believe that our techniques can be extended to further classes of dynamic graphs
or even in other domains of theoretical computer science, such as distributed graph algorithms or streaming algorithms.
One further interesting implication of our work is presenting the limitations of dynamic graph algorithms on expanders, thus complementing recent algorithmic results that use expander decompositions in dynamic graphs.

\begin{table}[t]
    \newcommand{\partialline}{\cline{2-3}}
    \small
    \renewcommand{\arraystretch}{1.1}
    \centering
    \begin{tabular}{|c|c|c|c|c|}
    \hline
    \textbf{Problem} & \textbf{Class} & \textbf{Section} & $u(m,N)$ & $q(m,N)$ \\ \hline
    \multirow{5}{*}{Maximum Matching} & $\Delta \le 3$ & \Cref{ssec:match_const_degree} & \multirow{5}{*}{\eb{1/2}} & \multirow{5}{*}{\eb{1}} \\\partialline
                  & constant degree \& expansion & \Cref{ssec:match_expand_proof} & & \\\partialline
                  & power-law graphs & \Cref{ssec:match_powerlaw} & & \\\partialline
                  & $\Delta \le 3$, partially dynamic & \Cref{ssec:match_partial} & & \\\cline{2-5}
                  & $\Delta \le N^t$ & \Cref{ssec:match_vary_degree} & \ebn{(1+t)/2} & \ebn{1+t} \\
    \hline
\multirow{6}{*}{$(s,t)$-distance} & $\Delta \le 3$ & \Cref{ssec:st_const_degree} & \multirow{6}{*}{\eb{1/2}} & \multirow{6}{*}{\eb{1}} \\\partialline
                                  & $(3-\delta)$-approx, $\Delta \le 3$ & \Cref{ssec:st_const_apx} & & \\\partialline
                                  & constant degree \& expansion & \Cref{ssec:st_expand} & & \\\partialline
                                  & power-law graphs & \Cref{ssec:st_powerlaw} & & \\\partialline
                                  & $\Delta \le 3$, partially dynamic & \Cref{ssec:st_partial} & & \\\cline{2-5}
                  & $\Delta \le N^t$ & \Cref{ssec:st_vary_degree} & \ebn{(1+t)/2} & \ebn{1+t} \\
    \hline
\multirow{3}{*}{Densest Subgraph} & $\Delta \le 5$ & \Cref{ssec:dense_const_degree} & \multirow{3}{*}{\ebn{1/4}} & \multirow{3}{*}{\ebn{1/2}} \\\partialline
    & constant degree \& expansion & \Cref{ssec:dense_expand} & & \\ \partialline
    & power-law graphs & \Cref{ssec:dense_powerlaw} & & \\
    \hline
    \end{tabular}
    \caption{Our results for graphs on $N$ nodes with $m$ edges. For every $u$ and $q$ stated above, there is no algorithm for the corresponding problem with amortized $O(u(m,N))$ update time and $O(q(m,N))$ query time simultaneously unless the \omv{} conjecture is false. The first three rows hold also for perfect matching. All the lower bounds in the table except for densest subgraph match the general \omv{} lower bounds\footnotemark.}
    \label{tab:our_results}
\end{table}

\subsection{Our Techniques}
\label{subsec:Techniques}
We prove lower bounds by reductions from the online matrix vector (\omv) conjecture~\cite{HenzingerKNS15}.
In these reductions, the input of an online problem, which is an $n\times n$ matrix $M$ and a sequence of $n$ pairs $(u,v)$ of $n$-vectors, is translated into a dynamic graph.
The reduction is built so that there exists a pair $(u,v)$ satisfying $uMv=1$ if and only if the dynamic graph has some desired property at some point of time.
While we follow the general framework of \omv{} lower bounds, the details are delicate,
as the dynamic graphs we construct should fall into specific graph classes at all times, while still maintaining the graph property under consideration.
We give a high-level overview of our reductions below.

One way to turn known \omv{}-to-dynamic graphs reductions into reductions that produce bounded-degree graphs is by replacing high-degree nodes by bounded-degree trees.
This technique has a rather clear and straightforward effect on the distances in the graph, so it is applicable when considering distance-related problems.
This, however, is far from being the case when considering other problems, such as maximum matchings.
Here, replacing a high-degree node with a gadget could adversely affect the desired matching size, since the gadget might create several augmenting paths that would not have existed when it was a single high-degree node.
To overcome this, we limit the possible maximum matching sizes, by designing a reduction graph with bounded-degree gadgets composed of paths, where the maximum matching is always either a perfect matching, or a near-perfect matching,
i.e., the matching size is either $N/2$ or $N/2-1$.
This reduction thus involves a large matching and a small gap between the $uMv=0$ and $uMv=1$ cases, and hence cannot be extended to achieve a lower bound for the approximation of the maximum matching size.
While this might seem as a limitation of our construction, recall that this is actually not the case: As described above, for any small $\delta > 0$ there is a constant time $(1+\delta)$-approximation dynamic algorithm for the problem, and, thus, such a lower bound cannot exist.

An even more delicate reduction we present is for proving a lower bound on the densest-subgraph problem.
A straightforward reduction would change $O(n)$ graph-edges for every bit of the input, which will allow us to make sure that the density of the densest subgraph changes by a significant amount when $uMv = 0$ versus when $uMv = 1$.
However, this would involve $O(n^2)$ updates for each $(u,v)$ input pair, and the reduction would fail to yield any non-trivial lower bound.
Thus, we are forced to change very few edges for each input bit, which renders an almost negligible effect on the density, making it difficult to control the exact density of the densest subgraph.
Our reduction balances these two factors, using a construction where each gadget is a sufficiently dense regular graph, while having each bit of the input translate into the existence or nonexistence of merely two edges inside specific gadgets.
As in the case of matchings, our lower bounds cannot be extended to approximations, as for any $\delta > 0$ there exists a fast algorithm with polylogarithmic update time for computing $(1-\delta)$-approximations to the densest subgraph.

We then extend these reductions from bounded-degree graphs to constant-degree \emph{constant-expansion} graphs.
First, the standard lower bound reductions contain sparse cuts if the inputs $M,u$ or $v$ are sparse, making a standard reduction graph far from being an expander.
Thus, we have to augment the graph with many more edges to make sure that it has no sparse cuts regardless of $M,u$ and $v$.
We do this augmentation ``inside a layer'' to prevent the additional edges from creating undesired short paths between $s$ and $t$, or spurious augmenting paths in the case of matchings.
Sparse cuts also exist in parts of the graph that do not depend on $M,u$ or $v$, and to handle these, we add edges of a constant-degree expander between a well-chosen set of nodes, thus guaranteeing the expansion without changing the required graph property.
Finally, in the case of distance-related problems, we note that expander graphs can have at most logarithmic diameter, but the substitution of nodes by trees described above increases the diameter to be at least logarithmic, leaving only a very small slack for our construction.

When studying densest subgraphs on expanders, adding edges in order to avoid sparse cuts might change the location and structure of the densest subgraph in an undesired way.
In order to guarantee the expansion in this case, we add a copy of all the graph nodes, build a constant-degree expander on the copies of the nodes, and then connect each node to its copy by a matching.

In dynamic power-law graphs where the node degrees may depend on the inputs $u,M,v$ and change over time, we have to guarantee that the degree changes incurred by the processing of different inputs do not cause a violation of the power-law distribution of degrees.
As before, all the changes must also be done without changing the graph property under consideration, and without performing too many update operations.
We address this problem by inserting or deleting edges in an online fashion in other parts of the reduction graph, to compensate for the changes incurred by processing the input vector pairs.

\subparagraph{Organisation}
\label{ssec:Organization}
\Cref{sec:Preliminaries} has notation and definitions.
\Cref{sec:Matching_LBs} presents the dynamic maximum matching lower bounds,
\Cref{sec:full_st_LBs} presents the dynamic $(s, t)$-distance lower bounds, and
\Cref{sec:full_Densest_Subgraph_LBs} presents the dynamic densest subgraph lower bounds.
The lower bounds for the partially dynamic setting are deferred to the appendix.

\section{Preliminaries}
\label{sec:Preliminaries}

Throughout the paper, we consider vectors and matrices that are boolean, and so a vector-matrix-vector multiplication outputs a single bit.
Henzinger \textit{et al.}~\cite{HenzingerKNS15} define the \emph{Online Matrix Vector} (\omv{}) and the \emph{Online Vector Matrix Vector} (\oumv{}) multiplication problems.

\begin{definition}[Online Matrix Vector Multiplication]
Let $M$ be a boolean $n \times n$ matrix. Preprocessing the matrix is allowed. Then, $n$ vectors $v^1, v^2, \ldots, v^n$ arrive one at a time, and the task is to output the product $Mv^i$ before the next vector is revealed.
\end{definition}

\begin{definition}[Online Vector Matrix Vector Multiplication]
Let $M$ be a boolean $n \times n$ matrix. Preprocessing the matrix is allowed. Then, $n$ vector pairs $(u^1, v^1), (u^2, v^2), \ldots, (u^n, v^n)$ arrive one at a time, and the task is to output the bit $u^iMv^i$ before the next vector pair is revealed.
\end{definition}
In their paper, they show that the \oumv{} problem can be reduced to the \omv{} problem, and conjecture that there is no truly subcubic time algorithm for \omv{}.

\begin{restatable}[\omv{}]{conjecture}{OMv}
\label{con:omv}
There is no algorithm for the \omv{} (and hence the \oumv{}) problem running in time $O(n^{3-\epsilon})$ for any constant $\epsilon >0$.
\end{restatable}

We work with the \oumv{} problem for all the reductions in our paper. We denote the length of our input vectors $u^i, v^i$ by $n$, and thus the matrix $M$ is of dimension $n \times n$.
We use upper indices to indicate the vector's location in the stream, but usually focus on one pair $(u,v)$ omitting these indices.
We use lower indices for a location in the vector or matrix, e.g., $u_i$ and $M_{ij}$.
We use $N$ to denote the number of nodes in our reduction graph.

\begin{definition}[Expansion]
The expansion parameter of a graph $G=(V,E)$ is defined as
\[h = \min \left\{ \frac{|E(S,\overline{S})|}{|S|}\; \middle|\; \emptyset \neq S \subseteq V,\; |S| \le |V|/2 \right\}
\]
\end{definition}
where $E(S,\overline{S})$ is the number of edges from $S$ to $V\setminus S$. We call a graph with expansion $h$ a $h$-\emph{expander}. Works on dynamic algorithms use a different definition of expansion parameter $h'$, called \emph{volume expansion}.
However, when considering constant-degree graphs with constant expansion (as we do in this paper), both parameters are within a $\Delta$ factor of each other, so we only consider the expansion parameter $h$ in our proofs.

We study \emph{power-law} graphs as introduced in~\cite{Aiello2001}, and only consider the setting where $\beta > 2$. In the following definition, if the number of nodes $N$ in the graph is fixed, then we get that $\alpha$ is roughly $\ln N$.

\begin{definition}[Power Law]
A graph is said to follow an $(\alpha, \beta)$-power law distribution if the number $N_d$ of nodes with degree $d$ is inversely proportional to $d^{\beta}$ for some constant $\beta > 0$. That is,
\[
N_d = \left\lfloor \frac{e^{\alpha}}{d^{\beta}} \right\rfloor \approx \pl{\beta}
,\]
where $\zeta(\beta) = \sum_{i=1}^{\infty} 1/i ^\beta$ is the Riemann Zeta function.
\end{definition}

Since dynamic graphs allow edge insertions and deletions, it is impossible to maintain an exact degree distribution at all times. Hence, we introduce the notion of approximate power-law distributions to afford some slack for dynamic changes. One natural relaxation is to allow $\beta$ to vary within an interval.
\begin{definition}[$\beta$-Varying Power Law]
A graph is said to follow an $(\alpha, \beta_1, \beta_2)$-varying power law distribution if the number $N_d$ of nodes with degree $d$ satisfies
\[
\min \left\{ \pl{\beta_1}, \pl{\beta_2} \right\} \le
N_d \le
\max \left\{ \pl{\beta_1}, \pl{\beta_2} \right\}
,\]
\end{definition}

This relaxation of an exact power law, while being natural, is a global relaxation rather than a local one. Thus we also define two locally approximate definitions below that allow similar slack for all degrees.
\begin{definition}[Additively Approximate Power Law]
A graph is said to follow an $(\alpha, \beta, c)$-additively approximate power law distribution if the number $N_d$ of nodes of degree $d$ for a realisable degree $d$ satisfies \[
\pl\beta - c \le
N_d \le \pl\beta + c
\]
where we say that $d$ is a realisable degree if there is a node of degree $d$ in an $(\alpha, \beta)$-power law graph.
\end{definition}

\begin{definition}[Multiplicatively Approximate Power Law]
A graph is said to follow an $(\alpha, \beta, \epsilon)$-multiplicatively approximate power law distribution if the number $N_d$ of nodes of degree $d$ satisfies
\[
 \frac{1}{(1+\epsilon)} \cdot \pl\beta \le
 N_d \le (1+\epsilon) \cdot \pl\beta
\]
\end{definition}

Our lower bounds contain at most four nodes that are one degree away from an exact power-law distribution, and thus hold in all the models discussed above with any reasonable parameter regime. We note a couple of properties of power-law graphs that we use in our lower bounds.

The maximum realizable degree in a power law graph is $\left\lceil e^{\alpha/\beta} \right\rceil$, since \[
N_d \ge 1 \iff \left\lfloor \frac{e^{\alpha}}{d^{\beta}} \right\rfloor \ge 1 \iff d \le \left\lceil e^{\alpha/\beta} \right\rceil
.\] In terms of $N$, the maximum degree in a power-law graph is $\Delta = \left\lceil \left( N/\zeta(\beta) \right)^{1/\beta} \right\rceil < \sqrt{N} $.

We work in the setting where $\beta > 2$. Note that in this setting, the number of edges in the graph is given by \[
|E| = \frac{1}{2} \cdot \frac{\zeta(\beta-1)}{\zeta(\beta)} \cdot N
,\] which is linear in the number of nodes for a fixed $\beta$.

\section{Lower Bounds for Dynamic Maximum Matching}
\label{sec:Matching_LBs}

In this section, we present our lower bound results for the maximum cardinality matching problem. The previous matching lower bounds on general graphs~\cite{HenzingerKNS15, dahlgaardhardness2016} use reduction graphs that contain nodes with degree $\Omega(N)$. Towards showing a lower bound on expanders, we first sparsify the original reduction.

In \Cref{ssec:match_const_degree}, we give a lower bound for maximum matching on graphs with maximum degree $3$. In \Cref{ssec:match_vary_degree}, we show that the distinction between the unbounded and constant-degree reductions is not discrete, by giving a lower bound reduction parameterized on the maximum degree allowed in the graph. In \Cref{ssec:match_expand_proof}, we give a reduction graph that has constant expansion. Finally, we show our lower bounds for power-law graphs in \Cref{ssec:match_powerlaw}.

\subsection{Constant-Degree Graph}
\label{ssec:match_const_degree}

\begin{minipage}{.55\textwidth}
We first perform a simple reduction that shows that maintaining maximum matchings is hard even on graphs where the maximum degree is $3$.
We use the following gadget composed of paths to maintain matching properties in our reduction graph during sparsification---see \cref{fig:matching_gadgets}.
Our gadget construction starts by replacing each node of a dense reduction by a path; we refer to each path as a `subgadget'.
Connecting every node of this new subgadget with nodes outside the subgadget might create unwanted matchings of larger sizes, so instead we carefully choose a subset of the path nodes to connect outside the subgadget.
\end{minipage}
\hfill
\begin{minipage}{0.42\textwidth}
    \centering
    \captionsetup{type=figure}
    \includegraphics[scale=0.12,trim={24cm 2cm 0cm 0cm}, clip]{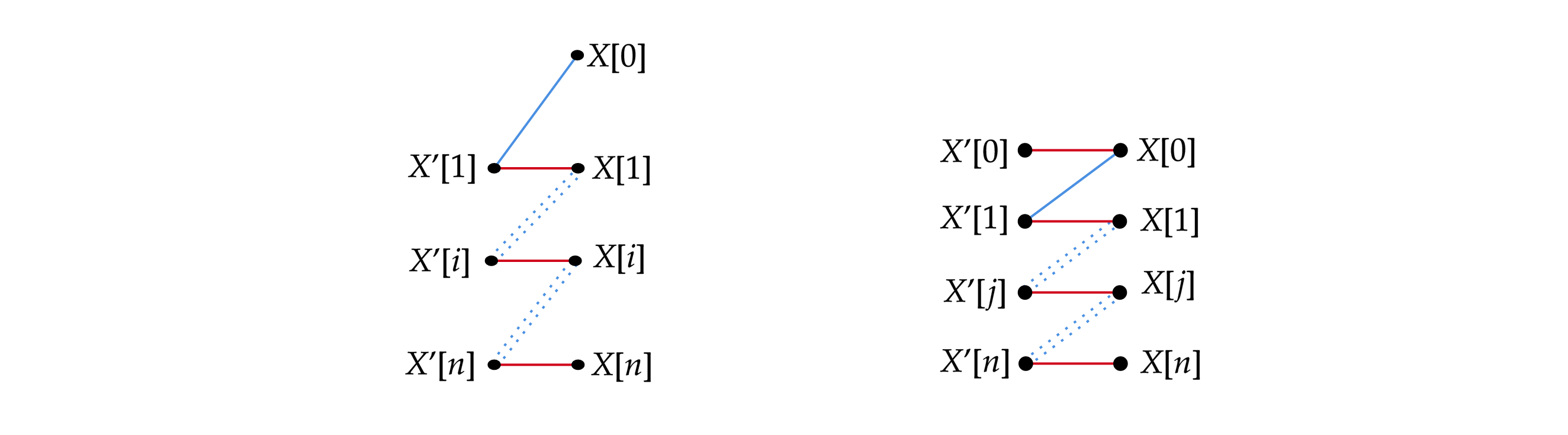}
    \captionof{figure}{Odd and even sized paths used in the maximum matching lower bounds. The canonical matchings are marked in red.}
    \label{fig:matching_gadgets}
\end{minipage}
\\\\

\cref{fig:matching_gadgets} shows odd and even paths (``odd'' and ``even'' describe the number of nodes) with a ``canonical'' matching for each of them marked in red. Next, we detail the connections outside the subgadgets.

Consider an odd path on $2n+1$ vertices, and a bipartition of the vertices into $(X', X)$ with $|X'| \le |X|$. Indexing the vertices as $X[0]$ and $ X'[i], X[i]$ for $1 \le i \le n$, our canonical matching matches $X[i]$ with $X'[i]$, and leaves $X[0]$ unmatched. We connect only the vertices $\{ X[i] \mid 1 \le i \le n \}$ outside the subgadget, while vertices in $X'$ and $X[0]$ only have edges inside the subgadget.
For an even path on $2n+2$ vertices indexed as above, our canonical matching is perfect, and matches $X[i]$ to $X'[i]$. Only vertex $X'[0]$ and all the vertices in $X$ are connected outside the subgadget, and all vertices $X'[i]$, $1 \le i \le n$, only have edges within the subgadget.

\begin{definition}[Reduction gadget]
A \emph{reduction gadget} with $x$ subgadgets of size $y$ is a bipartite graph composed of $x$ subgraphs, each of which is a path on $y$ nodes.
\end{definition}

\paragraph*{Reduction Graph}
\label{sssec:match_static}

The reduction graph is composed of two odd-sized reduction gadgets, and two even-sized reduction gadgets as follows.
\begin{itemize}
	\item A reduction gadget with one subgadget of size $2n+1$, on a set $L_1 \cup L_2$. The nodes are labelled $L_1[i]$ for $1 \le i \le n$, and $L_2[i]$ for $0 \le i \le n$, and the path is from $L_2[0]$ to $L_2[n]$.

	\item A reduction gadget with $n$ subgadgets of size $2n+2$ each, on a set $L_3 \cup L_4$. The subgadgets are labelled $LG[i]$ for $1 \le i \le n$, and the nodes of subgadget $LG[i]$ are labelled $L_3[i,j]$ or $L_4[i,j]$ for $0 \le j \le n$ depending on whether the node is in $L_3$ or $L_4$. The path in each subgadget goes from $L_3[i,0]$ to $L_4[i,n]$.

    \item A copy of the above structure, with node sets marked $R_1,R_2,R_3,R_4$ instead of $L_1,L_2,L_3,L_4$, respectively.

    \item For the matrix $M$, add the edge $(L_4[i,j],R_4[j,i])$ if $M_{ij}=1$.

    \item Given an input vector $u$, for each $i\in[n]$, add the edge $(L_2[i], L_3[i])$ if $u_i=1$.

    \item Given an input vector $v$, for each $j\in[n]$, add the edge $(R_2[j], R_3[j])$ if $v_j=1$.
\end{itemize}

The total number of nodes in the reduction graph is $N = 4n^2 + 8n + 2 = \Theta(n^2)$.
Note that our reduction graph is bipartite, and thus our lower bounds in this section hold for maintaining an exact maximum matching even on bipartite graphs.

\paragraph*{Matchings in the Graph}
\label{sssec:match_matchings}

We start by defining a base matching $B$ on the graph, which is made up of the canonical matchings on each of the gadgets. On the left side, $B$ matches $L_3[i,j]$ to $L_4[i,j]$, and $L_1[i]$ to $L_2[i]$ for all $i, j$. The matching on the right side is similar. Note that this matching always exists regardless of the input, and only $L_2[0]$ and $R_2[0]$ are unmatched in the entire graph. Thus $| B | = \frac{N}{2} - 1$.
We claim that this graph has a perfect matching if and only if $uMv = 1$. Let $C$ denote the maximum cardinality matching.

\begin{lemma}
\label{lem:constMatchSizeLem}
If $uMv = 1$, then $| C | = \frac{N}{2}$, and otherwise $| C | = \frac{N}{2} - 1$.
\end{lemma}

\begin{proof}
Since $B$ is always a matching of size $\frac{N}{2} - 1$ regardless of the input, the claim is equivalent to showing that $uMv = 1$ if and only if there is an augmenting path with respect to the matching $B$.

($\implies$) Suppose that $uMv = 1$, with $u_i = M_{ij} = v_j = 1$.
Consider the path $P$ composed of the following subpaths, of which all except $P_4$ start with an unmatched edge and end with a matched edge,
while $P_4$ both starts and ends with an unmatched edge.
\begin{itemize}
    \item $P_1 = L_2[0], L_1[1], L_2[1], \ldots, L_2[i]$
    \item $P_2 = L_2[i], L_3[i,0], L_4[i,0], \ldots, L_4[i,j]$
    \item $P_3 = L_4[i,j], R_4[j, i], R_3[j, i], \ldots, R_3[j, 0]$
    \item $P_4 = R_3[j,0], R_2[j], R_1[j], \ldots, R_2[0]$
\end{itemize}
Thus, $P$ is an augmenting path to the base matching $B$, which gives us that the maximum matching $C$ has to have size $> \frac{N}{2} - 1$, implying that the maximum matching $C$ is a perfect matching.

($\impliedby$) Suppose now that there exists an augmenting path $P$ to the base matching $B$ that starts at $s = L_2[0]$ and ends at $t = R_2[0]$.
\begin{itemize}
    \item Since $(\cup_i L_i, \cup_j R_j)$ is an $(s, t)$-cut, there is at least one crossing edge, say $(L_4[i,j], R_4[j,i])$, in $P$. Thus $M_{ij}$ = 1.

    \item Since $P$ leaves the subgadget $LG[i]$ using $(L_4[i,j], R_4[j,i])$, it should have entered the subgadget at some previous instance. Since $(L_4[i,j], R_4[j,i])$ is an unmatched edge and all the matching edges in $LG[i]$ are within the subgadget, $P$ should have entered the subgadget using an unmatched edge. As all the matching edges in $LG[i]$ are between $L_3$ and $L_4$, $P$ cannot both enter and exit the subgadget through $L_4$. Thus $P$ enters $LG[i]$ through $L_3$. However, the only possible unmatched edge from $L_3$ leaving the subgadget is the edge $(L_3[i,0], L_2[i])$. Thus $P$ uses the edge $(L_3[i,0], L_2[i])$ to enter the subgadget $LG[i]$, and so $u_i = 1$.

    \item The path $P$ now enters the subgadget $RG[j]$ through the unmatched edge $(L_4[i, j], R_4[j,i])$. As before, all the matched edges in $RG[j]$ are between $R_4$ and $R_3$, and so $P$ has to exit the subgadget using an unmatched edge from $R_3$. However, the only possible unmatched edge from $R_3$ leaving the subgadget is the edge $(R_3[j,0],R_2[j])$. Thus the edge $(R_3[j, 0], R_2[j])$ is used by $P$, giving us that $v_j = 1$.
\end{itemize}
This gives us that $uMv = 1$ as required.
\end{proof}

\paragraph*{Complexity of the Reduction}
\label{sssec:match_complexity}
We are now ready to prove the theorem.

\begin{restatable}{theorem}{matchConst}
	\label{thm:matchConst}
	For any constant $\epsilon>0$,
	there is no dynamic algorithm maintaining a maximum matching or determining the existence of a perfect matching,
	on all $N$-node graphs with maximum degree $\Delta \le 3$,
	with amortized \ub{1/2} update time and \ub{1} query time,
	unless the \omv{} conjecture is false.
\end{restatable}

\begin{proof}
Consider the reduction graph above. It consists of $N = 4n^2 + 8n + 2 = \Theta(n^2)$ nodes. Every time we get a new $(u,v)$ input vector pair, we delete all the edges between $L_2 \times L_3$ and $R_2 \times R_3$ and insert edges according to the new input vectors. This takes $O(n)$ updates in total. After that, we query once for the size of the maximum matching in this new graph, and return $1$ if and only if $|C| = \frac{N}{2}$.

Thus for each pair of input vectors, we perform $O(n)$ updates and $O(1)$ query. In total, checking $n$ pairs takes us $O(n^2)$ updates and $O(n)$ query. If there were an algorithm for maximum matching on constant-degree graphs with update time \ub{1/2} (i.e., $O(n^{1-2\epsilon})$) and query time \ub{1} (i.e., $O(n^{2-2\epsilon})$), then we can decide if $uMv=1$ for all $n$ pairs in $O(n^{3-2\epsilon})$ time, contradicting the \omv{} conjecture.
\end{proof}

\subsection{Varying Degree Graph}
\label{ssec:match_vary_degree}

We present a reduction that gives a lower bound parameterized on the maximum degree of the graph.
Note that setting $t=1$ and $t=0$ in the following theorem give us the unbounded degree lower bound of~\cite{HenzingerKNS15} and \Cref{thm:matchConst} respectively.

\begin{restatable}{theorem}{matchVary}
	\label{thm:matchVary}
	For any $0\leq t\leq 1$ and any constant $\epsilon>0$,
	there is no dynamic algorithm maintaining a maximum matching  or determining the existence of a perfect matching,
	on all $N$-node graphs with maximum degree $\Delta = O(N^t)$,
	with amortized \ub{\frac{1+t}{2}} update time and \ub{1+t} query time,
	unless the \omv{} conjecture is false.
\end{restatable}

Given $0 \le t \le 1$, we construct a reduction graph that has maximum degree $O(N^t)$. $L_1, L_2, R_1, R_2$ and the edges for $u$ and $v$ are the same as in the previous construction. We now detail the changes for $L_3, L_4, R_3, R_4$ and the edges that depend on $M$.
\begin{itemize}
	\item $LG$ is now an even reduction gadget with $n$ subgadgets of size $2 n^{\frac{1-t}{1+t}} + 2$ each.
	The path in each subgadget goes from $L_3[i,0]$ to $L_4[i,n^{\frac{1-t}{1+t}}]$, and similarly for $R_3$ and $R_4$.

    \item
    For the matrix $M$, if $M_{ij}=1$, let $i' = \left\lceil i \cdot n^{-2t/(t+1)} \right\rceil$ and $j' = \left\lceil j \cdot n^{-2t/(t+1)} \right\rceil $
    and add the edge $(L_4[i,j'],R_4[i',j])$ to the graph.
\end{itemize}

Note that the augmenting paths in this reduction graph are the same as in the constant-degree reduction graph by a similar proof as in~\Cref{lem:matchSizeLem}. By construction, each node in $L_4$ is connected to at most $n^{2t/(t+1)}$ nodes in $R_4$, and each node in $R_4$ is connected to at most $n^{2t/(t+1)}$ nodes in $L_4$.
The proof of the theorem is similar to the proof of \cref{thm:matchConst}.

\begin{proof}
The number of nodes in the reduction graph described in \cref{ssec:match_vary_degree} is dominated by the number of nodes in $L_4$ and $R_4$. Thus the total number of nodes in the reduction graph is $N = \Theta(n^{2/(t+1)})$ nodes. Each node in $L_4$ and $R_4$ has at most $n^{2t/(t+1)}$ edges of $M$ incident on it by a similar argument as in the proof of \Cref{thm:stVary}. Thus the maximum degree in the graph is $O(n^{2t/(t+1)}) = O(N^t)$ as required. The rest is similar to \Cref{thm:matchConst}.

Every time we get a new $(u,v)$ input vector pair, we delete all the edges between $L_2 \times L_3$ and $R_2 \times R_3$ and insert edges according to the new input.
Thus, for each pair of input vectors, we perform $O(n)$ updates and $O(1)$ queries. In total, checking $n$ pairs takes us $O(n^2)$ updates and $O(n)$ queries. If there were an algorithm for maximum matching on graphs with maximum degree bounded by $N^t$ with update time \ub{\frac{1+t}{2}} (i.e., $O(n^{1-2\epsilon})$) and query time \ub{1+t} (i.e., $O(n^{2-2\epsilon})$), then we can decide if $uMv=1$ for all $n$ pairs in $O(n^{3-2\epsilon})$ time, contradicting the \omv{} conjecture.
\end{proof}

\subsection{Expander Graph}
\label{ssec:match_expand_proof}

The previous matching lower bounds on general graphs~\cite{HenzingerKNS15, dahlgaardhardness2016} use reduction graphs that contain nodes with degree $\Omega(N)$. In this section, we construct a constant-degree reduction graph with constant expansion.
\paragraph*{Reduction gadgets}
While the reduction gadgets in \Cref{ssec:match_const_degree} suffice for sparsification, we need additional constructions in order to guarantee constant expansion. In particular, it turns out that adding edges inside a subgadget does not suffice for constant expansion, and we are forced to add edges between subgadgets. Our construction adds edges on the same side of the bipartition across subgadgets, and our proof implicitly shows that if the newly added edges take part in any augmenting path, then there also exists an augmenting path in the subgraph devoid of any newly inserted edge.

The reduction graph consists of a left subgraph ($L$) and a right subgraph ($R$), connected together by edges corresponding to the matrix $M$.
Note that for constant expansion, we need the number of edges of $M$ to be a constant fraction of the sizes of $L$ and $R$. While it would be possible to construct a reduction graph with $|L|$ and $|R|$ that depend on the input matrix $M$, we instead choose to augment the input matrix and vectors as it simplifies notation.
We thus augment the input beforehand to ensure that there are $\Omega(n^2)$ edges crossing from $L$ to $R$. To this end, we work with the vectors $\hat{u} = (u\,\, 0)$ and $\hat{v} = (v\,\, 0)$ of dimension $2n$, and the matrix $\hat{M} = \begin{psmallmatrix} M & 1 \\ 1 & 1 \end{psmallmatrix} $ of dimension $2n \times 2n$. It is easy to see that $uMv = 1 \iff \hat{u}\hat{M}\hat{v} = 1$.

\begin{definition}[Reinforced gadget]
A \emph{reinforced gadget} with $x$ subgadgets of size $y$ consists of $x$ subgraphs, each of which is a path on $y$ nodes. The nodes are bipartitioned into sets $(X', X)$ with the larger side of the partition labelled as $X$ in each subgadget. Thus $|X'| \le |X|$. It is then augmented with the following edge-set:
Consider a degree-$d$ expander graph on $x \cdot \left\lceil \frac{y}{2} \right\rceil $ nodes, choose an arbitrary bijection between the expander nodes and $X$, and add the expander edges to these nodes accordingly.
The resulting graph is the reinforced gadget.
\end{definition}

Note that reinforced gadgets are not bipartite. Thus, while the constant-degree lower bounds hold for bipartite matching, the expander result is for maximum matching on general graphs. In what follows, we drop the hats from $\hat{u}, \hat{M}$ and $\hat{v}$ for simplicity, but continue our analysis with their dimensions as $2n, 2n \times 2n, 2n$ respectively.

\paragraph*{Reduction Graph}

We use the following reduction graph, composed of two odd-sized reinforced gadgets and two even-sized reinforced gadgets.
\begin{itemize}
	\item A reinforced gadget with one subgadget of size $4n+1$, on a set $L_1 \cup L_2$. The nodes are labelled $L_1[i]$ for $1 \le i \le 2n$, and $L_2[i]$ for $0 \le i \le 2n$. The path is from $L_2[0]$ to $L_2[2n]$, and the expander is on $L_2$.

	\item A reinforced gadget with $2n$ subgadgets of size $4n+2$, on a set $L_3 \cup L_4$. The subgadgets are labelled $LG[i]$ for $1 \le i \le 2n$, and the nodes of subgadget $LG[i]$ are labelled $L_3[i,j]$ or $L_4[i,j]$ for $0 \le j \le 2n$ depending on whether the node is in $L_3$ or $L_4$. The path in each subgadget goes from $L_3[i,0]$ to $L_4[i,2n]$, and the expander is on $L_4$.

    \item A copy of the above structure, with node sets marked $R_i$ instead of $L_i$, respectively.

    \item For the matrix $M$, add the edge $(L_4[i,j],R_4[j,i])$ if $M_{ij}=1$.

    \item Given an input vector $u$, for each $i\in[2n]$, add the edge $(L_2[i], L_3[i,0])$ if $u_i=1$, and $(L_2[i], L_4[i,0])$ otherwise.

    \item Given an input vector $v$, for each $j\in[2n]$, add the edge $(R_2[j], R_3[j,0])$ if $v_j=1$, and add the edge $(R_2[j], R_4[j,0])$ otherwise.
\end{itemize}

The total number of nodes in the reduction graph is $N = 16n^2 + 16n + 2 = \Theta(n^2)$.

\begin{figure}
    \centering
    \includegraphics[scale=0.15]{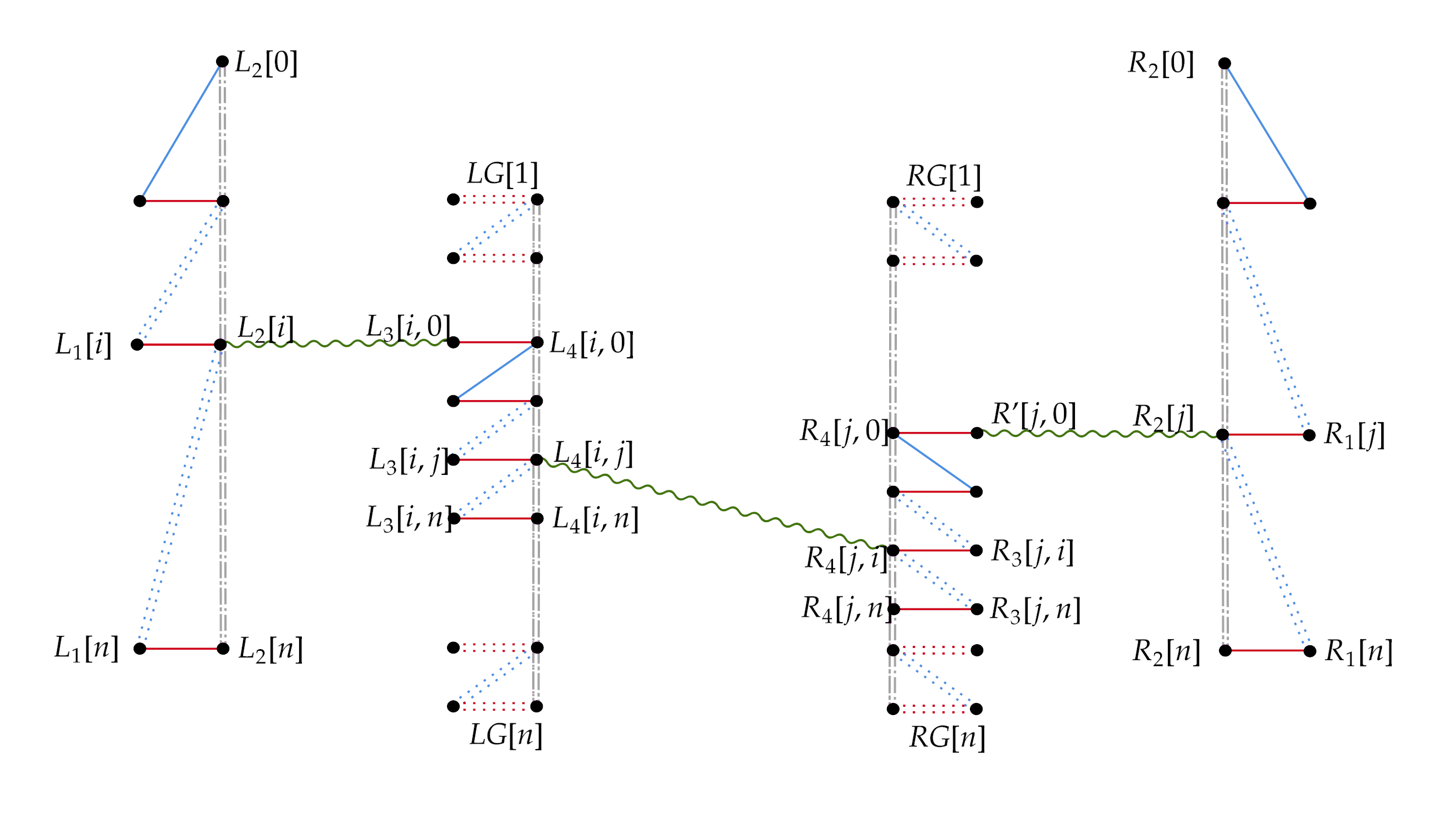}
    \caption{The expander reduction graph for maximum matching. The red lines denote the canonical matching, the blue lines denote the paths in each subgadget, the grey lines denote the expander edges, and the green lines denote the input-dependent edges.}
    \label{fig:matching-reduction-gadget}
\end{figure}

\paragraph*{Matchings in the Graph}

Our base matching $B$ is similar to the constant-degree case, and is made up of the canonical matchings on each of the gadgets. On the left side, $B$ matches $L_3[i,j]$ to $L_4[i,j]$, and $L_1[i]$ to $L_2[i]$ for all $i, j$. The matching on the right side is similar. Note that this matching always exists regardless of the input, and only $L_2[0]$ and $R_2[0]$ are unmatched in the entire graph. Thus $| B | = \frac{N}{2} - 1$.
We claim that this graph has a perfect matching if and only if $uMv = 1$. Let $C$ denote the maximum cardinality matching.

\begin{lemma}
\label{lem:matchSizeLem}
If $uMv = 1$, then $| C | = \frac{N}{2}$, and otherwise $| C | = \frac{N}{2} - 1$.
\end{lemma}

The proof of \Cref{lem:constMatchSizeLem} without any modifications shows \Cref{lem:matchSizeLem} even with the additional edges present.

\paragraph*{Complexity of the Reduction}

On the arrival of a new vector pair, we first add all the edges corresponding to the new input (if they do not already exist), and then remove the previous vector pair's edges, as opposed to the usual convention of first deleting the previous edges and inserting the new ones. This ensures that the graph remains an expander at all steps.
Assuming that we have proved the expansion properties required, since number of edges in a constant-degree graph is $O(N)$, we get the following theorem for expanders.

\begin{restatable}{theorem}{matchExpand}
	\label{thm:matchExpand}
	For any constant $\epsilon>0$,
	there is no dynamic algorithm maintaining a maximum matching  or determining the existence of a perfect matching,
	on all $N$-node graphs with constant degree and constant expansion,
	with amortized \ub{1/2} update time and \ub{1} query time,
	unless the \omv{} conjecture is false.
\end{restatable}

\begin{proof}
Consider the reduction graph above. It consists of $N = 16n^2 + 16n + 2 = \Theta(n^2)$ nodes. Every time we get a new $(u,v)$ input vector pair, we update $L_2 \times L_3$ and $R_2 \times R_3$ as detailed above. This takes $O(n)$ updates in total. After that, we query once for the size of the maximum matching in this new graph, and return $1$ if and only if $|C| = \frac{N}{2}$.

Thus for each pair of input vectors, we perform $O(n)$ updates and $O(1)$ query. In total, checking $n$ vector pairs takes us $O(n^2)$ updates and $O(n)$ query. If there were an algorithm for maximum matching on constant-degree graphs with update time \ub{1/2} (i.e., $O(n^{1-2\epsilon})$) and query time \ub{1} (i.e., $O(n^{2-2\epsilon})$), then we can decide if $uMv=1$ for all $n$ pairs in $O(n^{3-2\epsilon})$ time, contradicting the \omv{} conjecture.
\end{proof}

Let us now show that the graph described is indeed an $h_1$-expander graph, for some constant $h_1>0$.

\begin{restatable}{lemma}{matchExpandExpansion}
\label{lem:matchExpandExpansion}
The reduction graph has constant expansion.
\end{restatable}

Throughout the rest of this section, let $S$ be an arbitrary subset of vertices in the reduction graph, with $|S| < N/2 = 8n^2 + 8n + 1$. To simplify our proofs, we consider the reduction graphs with $n > 90$, since then $8n^2 + 8n + 1 < 8.1n^2$. We use ``$S$ \emph{expands}'' as a shorthand for $|E(S, \bar{S})| \ge c \cdot |S|$ for some constant $c > 0$.

Note that the node sets in the graph have the following sizes:
$|L_1|=2n$; $|L_2|=2n+1$; $|L_3|=|L_4|=4n^2+2n$, and $|R_i| = |L_i|$.
We divide the expansion proof into two sections based on whether its intersection with the middle layers ($L_4$ or $R_4$) is large or small.

We first deal with the case when the intersection is large.

\begin{restatable}{lemma}{matchExpandIntersectLarge}
\label{lem:matchExpandIntersectLarge}
If $| S \cap L_4 | \ge 0.1n^2 $ or $|S \cap R_4 | \ge 0.1 n^2$, then $S$ expands.
\end{restatable}

We encapsulate the proof ideas for this lemma in~\Cref{tab:matchLarge}, and split the proof into three lemmas for convenience.
First, we show that if its intersection with both the middle layers is very large, then $S$ expands. Our proof uses the canonical matching on $L_3 \cup L_4$.

\begin{table}[t]
    \renewcommand{\arraystretch}{1.2}
    \centering
    \begin{tabular}{|c|c|c|}
    \hline
    \textbf{Sizes} & \textbf{Proof ideas} & \textbf{Proof} \\ \hline
    $S_{L4} > 3.9 n^2$ $\wedge$ $S_{R4} > 3.9 n^2$ & Use the perfect matching on $L_3 \cup L_4$ & \Cref{lem:matchExpandBothLarge} \\ \hline
    $0.1 n^2 \le S_{L4} \le 3.9 n^2$ & Use the expander on $L_4$ & \Cref{lem:matchExpandIntersectMed} \\ \hline
    $S_{L4} > 3.9 n^2$ $\wedge$ $S_{R4} < 0.1 n^2$ & Use the edges of matrix $M$ & \Cref{lem:matchExpandOneLargeOneSmall} \\ \hline
    \end{tabular}
    \caption{Ideas for the proof of expansion in the matching reduction graph when either $S \cap L_4$ or $S \cap R_4$ is large. we use $S_{L4}$ as shorthand for $|S \cap L_4|$. The cases when $0.1 n^2 \le S_{R4} \le 3.9 n^2$ and $S_{R4} > 3.9 n^2$ $\wedge$ $S_{L4} < 0.1 n^2$ are symmetric to the ones presented in the table.}
    \label{tab:matchLarge}
\end{table}

\begin{restatable}{lemma}{matchExpandBothLarge}
\label{lem:matchExpandBothLarge}
If $|S \cap L_4| > 3.9 n^2$ and $|S \cap R_4| > 3.9 n^2$, then $S$ expands.
\end{restatable}
\begin{proof}
The two conditions together imply that $|S \cap L_3 | \le 0.4 n^2$. Thus
\begin{align*}
|E(S, \bar{S})|
&\ge |E(S\cap L_4, \bar{S} \cap L_3)| \\
&= |E(S\cap L_4, L_3)| - |E(S \cap L_4, S \cap L_3)| \\
&\ge 3.9 n^2 - |E(S \cap L_4, S \cap L_3)| \tag*{(by the canonical matching)}\\
&\ge 3.9 n^2 - 1.2n^2 \tag*{(since $\deg(v) \le 3$ for all $v \in L_3$)}\\
&\ge (2.7/8.1) \cdot |S| \tag*{(by upper bound on $|S|$)}
\end{align*}
which proves our claim.
\end{proof}

Next, we show that if the intersection with one of the middle layers is of medium size, then $S$ expands. We prove this by using the fact that there is an $h_0$-expander on the middle layers.

\begin{restatable}{lemma}{matchExpandIntersectMed}
\label{lem:matchExpandIntersectMed}
If $0.1n^2 \le | S \cap L_4 | \le 3.9n^2$, then $S$ expands.
\end{restatable}
\begin{proof}
Let $T$ be the smaller of the two sets $S \cap L_4$ and $\bar{S} \cap L_4$, then $|T| \le |L_4|/2$.
Thus
\begin{align*}
    |E(S, \bar{S})|
    &\ge |E(S \cap L_4, \bar{S} \cap L_4)| \\
    &= |E(T, L_4 \setminus T)| \tag*{(by definition of $T$)} \\
    &\ge h_0 \cdot |T| \tag*{(by expansion)}\\
    &\ge h_0 \cdot 0.1 \cdot n^2 \tag*{(by constraint on $|S \cap L_4|$)}\\
    &\ge h_0 \cdot (0.1/8.1) \cdot |S| \tag*{(by upper bound on $|S|$)}
\end{align*}
which proves our claim.
\end{proof}

The above observation also holds when the size of $S \cap R_4$ is in the same range, using the expander on $R_4$. Finally, we show that if one of the intersections is large and the other is small, then $S$ expands. We use the edges of the matrix $M$ to do this.

\begin{restatable}{lemma}{matchExpandOneLargeOneSmall}
\label{lem:matchExpandOneLargeOneSmall}
If $|S \cap L_4| > 3.9 n^2$ and $|S \cap R_4| < 0.1n^2$, then $S$ expands.
\end{restatable}
\begin{proof}
Recall that our reduction graph has $\ge 3n^2$ edges crossing from $L_4$ to $R_4$. Thus
\begin{align*}
|E(S, \bar{S})|
&\ge |E(S\cap L_4, \bar{S} \cap R_4)| \\
&= |E(S\cap L_4, R_4)| - |E(S\cap L_4, S \cap R_4)| \\
&\ge |E(S\cap L_4, R_4)| - 0.1n^2 \tag*{(by constraint on $|S \cap R_4|$)}\\
&= |E(L_4, R_4)| - |E(\bar{S} \cap L_4, R_4)| - 0.1n^2 \\
&\ge 3n^2 - |E(\bar{S} \cap L_4, R_4)| - 0.1n^2 \tag*{(by construction)}\\
&\ge 3n^2 - 0.2n^2 - 0.1n^2 \tag*{(by constraint on $|S \cap L_4|$)}\\
&\ge (2.7/8.1) \cdot |S| \tag*{(by upper bound on $|S|$)}
\end{align*}
which proves our claim.
\end{proof}

A symmetric argument works by swapping $L_4$ and $R_4$ in the above proof. Lemmas~\ref{lem:matchExpandBothLarge}, \ref{lem:matchExpandIntersectMed}, and \ref{lem:matchExpandOneLargeOneSmall} together prove \Cref{lem:matchExpandIntersectLarge}. We are left with the case when the intersection of $S$ with both the middle layers is small.

\begin{restatable}{lemma}{matchExpandIntersectSmall}
\label{lem:matchExpandIntersectSmall}
If $| S \cap L_4 | < 0.1n^2 $ and $|S \cap R_4 | < 0.1 n^2$, then $S$ expands.
\end{restatable}

In this case, we show that the intersection of $S$ with the left side ($L = \cup_i L_i$) and the right side ($R = \cup_i R_i$) of the graph expand within their respective sides. This then proves expansion of $S$, since
\begin{align*}
|E(S, \bar{S})|
&\ge |E(S\cap L. \bar{S} \cap L)| + |E(S\cap R. \bar{S} \cap R)| \\
&\ge h_1 \cdot |S \cap L| + h_1 \cdot |S \cap R| \tag*{(by expansion within each side)} \\
&\ge h_1 \cdot |S|
\end{align*}

We now concentrate on proving the expansion of $S \cap L$ in $L$, since the right side expansion follows by similar arguments.

\begin{table}[t]
    \renewcommand{\arraystretch}{1.2}
    \centering
    \begin{tabular}{|c|c|c|}
    \hline
    \textbf{Sizes} & \textbf{Proof ideas} & \textbf{Proof} \\ \hline
    $S_{L3} > 2S_{L_4}$ $\wedge$ $S_{L3} > 0.2 n$ & Use the perfect matching on $L_3 \cup L_4$ & \Cref{lem:matchExpandThreeLarge} \\ \hline
    $S_{L3} \le 2S_{L_4}$ $\wedge$ $S_{L4} > 0.1 n$ & Use the expander on $L_4$ & \Cref{lem:matchExpandThreeSmall} \\ \hline
    \end{tabular}
    \caption{Ideas for the first part of the proof of expansion in the left side of the matching reduction graph when $S \cap L_4$ is small ($< 0.1 n^2$). Here we deal with the case when either $S_{L4} > 0.1n$ or $S_{L3} > 0.2n$. }
    \label{tab:matchSmall}
\end{table}

\begin{restatable}{lemma}{matchExpandThreeLarge}
\label{lem:matchExpandThreeLarge}
If $|S \cap L_3| > 2 |S \cap L_4|$ and $|S \cap L_3| > 0.2 n$, then $S$ expands.
\end{restatable}
\begin{proof}

First, we lower bound the number of crossing edges by $c \cdot |S \cap L_3|$.
\begin{align*}
|E(S \cap L, \bar{S} \cap L)|
&\ge |E(S\cap L_3, \bar{S} \cap L_4)| \\
&\ge |S \cap L_3| - |S \cap L_4| \tag*{(by the canonical matching)}\\
&\ge 0.5 \cdot |S \cap L_3| \tag*{(by constraint on $|S \cap L_4|$)}\\
\end{align*}
Then we lower bound $22 \cdot |S \cap L_3|$ by $|S \cap L|$, which proves our claim. The final inequality uses the constraint that $|S \cap L_3| > 0.2 n$.
\begin{align*}
22 \cdot |S \cap L_3|
&= |S \cap L_3| + |S \cap L_3| + 20 \cdot |S \cap L_3| \\
&\ge |S \cap L_3| + |S \cap L_4| + 20 \cdot |S \cap L_3| \tag*{(by constraint on $|S \cap L_4|$)}\\
&\ge |S \cap L_3| + |S \cap L_4| + |S \cap (L_1 \cup L_2)| \tag*{(using $|L_1 \cup L_2| = 4n+1$)}
\end{align*}
as required.
\end{proof}

\begin{restatable}{lemma}{matchExpandThreeSmall}
\label{lem:matchExpandThreeSmall}
If $|S \cap L_3| \le 2 |S \cap L_4|$ and $|S \cap L_4| > 0.1 n$, then $S$ expands.
\end{restatable}
\begin{proof}
Recall that we are working in the case when $|S \cap L_4| < 0.1 n^2 < |L_4|/2$.
\begin{align*}
|E(S \cap L, \bar{S} \cap L)|
&\ge |E(S\cap L_4, \bar{S} \cap L_4)| \\
&\ge h_0 \cdot |S \cap L_4| \tag*{(by expansion)}
\end{align*}
Lower bounding $|S \cap L_4|$ by $c \cdot |S \cap L|$ similar to~\Cref{lem:matchExpandThreeLarge} gives us the claim.
\end{proof}

Now we are only left with the case when $|S \cap L_4| < 0.1n$ and $|S \cap L_3| < 0.2n$.
In what follows, we use the bound below on $|S \cap L|$.
\begin{align*}
|S \cap L|
&= |S \cap (L_1 \cup L_2)| + |S \cap (L_3 \cup L_4)| \\
&\le 5n + |S \cap (L_3 \cup L_4)| \tag*{(by bound on $|L_1 \cup L_2|$)}\\
&\le 5.3n \tag*{(by constraint on $|S \cap (L_3 \cup L_4)|$)}
\end{align*}

\begin{table}[t]
    \renewcommand{\arraystretch}{1.2}
    \centering
    \begin{tabular}{|c|c|c|}
    \hline
    \textbf{Sizes} & \textbf{Proof ideas} & \textbf{Proof} \\ \hline
    $S_{L2} \ge 0.7 n$ & Use the edges of $u$ & \Cref{lem:matchExpandTwoLarge} \\ \hline
    $S_{L2} < 0.7 n$ $\wedge$ $S_{L1} \ge 1.3 n$ & Use the canonical matching on $L_1 \cup L_2$ & \Cref{lem:matchExpandOneLarge} \\ \hline
    $S_{L2} < 0.7 n$ $\wedge$ $S_{L1} < 1.3 n$ & Use the gadget expansion & \Cref{lem:matchExpandAllSmall} \\ \hline
    \end{tabular}
    \caption{Ideas for the second part of the proof of expansion in the left side of the matching reduction graph when $S \cap L_4$ is small. Here we deal with the case when $S_{L4} < 0.1 n$ and $S_{L3} < 0.2 n$.}
    \label{tab:matchVSmall}
\end{table}

\begin{restatable}{lemma}{matchExpandTwoLarge}
\label{lem:matchExpandTwoLarge}
If $|S \cap L_2| > 0.7 n$, then $S$ expands.
\end{restatable}
\begin{proof}
We use the fact that there is an edge from $L_2$ to $L_3 \cup L_4$ regardless of the value of $u$.
\begin{align*}
|E(S \cap L, \bar{S} \cap L)|
&\ge |E(S\cap L_2, \bar{S} \cap (L_3 \cup L_4))| \\
&= |E(S\cap L_2, L_3 \cup L_4)| - |E(S\cap L_2, S \cap (L_3 \cup L_4))| \\
&\ge 0.7n - |E(S\cap L_2, S \cap (L_3 \cup L_4))| \tag*{(by edges of $u$)} \\
&\ge 0.4n \tag*{(by constraint on $S \cap (L_3 \cup L_4)$)} \\
&\ge (0.4/5.3) \cdot |S \cap L| \tag*{(by constraint on $S \cap L$)}
\end{align*}
which proves our claim.
\end{proof}

\begin{restatable}{lemma}{matchExpandOneLarge}
\label{lem:matchExpandOneLarge}
If $|S \cap L_2| < 0.7n$ and $|S \cap L_1| \ge 1.3 n$, then $S$ expands.
\end{restatable}
\begin{proof}
We use the matching on $L_1 \cup L_2$, as follows
\begin{align*}
|E(S \cap L, \bar{S} \cap L)|
&\ge |E(S\cap L_1, \bar{S} \cap L_2)| \\
&\ge |S \cap L_1| - |S \cap L_2| \tag*{(by the canonical matching on $L_1 \cup L_2$)} \\
&\ge 0.6n \tag*{(by constraint on $L_1$ and $L_2$)} \\
&\ge (0.6/5.3) \cdot |S \cap L| \tag*{(by constraint on $S \cap L$)}
\end{align*}
which proves our claim.
\end{proof}

For the final case, we will need the following lemma about each reinforced gadget being an expander.

\begin{restatable}{lemma}{matchExpandGadget}
\label{lem:matchExpandGadget}
A reinforced gadget has constant expansion.
\end{restatable}

We use this lemma as follows.

\begin{restatable}{lemma}{matchExpandAllSmall}
\label{lem:matchExpandAllSmall}
If $|S \cap L_2| < 0.7 n$ and $|S \cap L_1| < 1.3 n$, then $S$ expands.
\end{restatable}
\begin{proof}
In this case, we reason about $L_1 \cup L_2$ and $L_3 \cup L_4$ separately.
Since the intersection of $S$ with each reinforced gadget covers less than half the nodes, i.e.,
\begin{align*}
|S \cap (L_1 \cup L_2)| &<\ 2n\ < |L_1 \cup L_2|/2, \\
|S \cap (L_3 \cup L_4)| &< 0.3n < |L_3 \cup L_4|/2,
\end{align*}
we use the expansion of each reduction gadget from \Cref{lem:matchExpandGadget} to get
\begin{align*}
|E(S \cap (L_1 \cup L_2), \bar{S} \cap (L_1 \cup L_2))|
&\ge c \cdot |S \cap (L_1 \cup L_2)|, \\
|E(S \cap (L_3 \cup L_4), \bar{S} \cap (L_3 \cup L_4))|
&\ge c \cdot |S \cap (L_3 \cup L_4)|,
\end{align*}
giving the expansion of $S \cap L$ as required.
\end{proof}

Lemmas \ref{lem:matchExpandThreeLarge}--\ref{lem:matchExpandAllSmall} together prove \Cref{lem:matchExpandIntersectSmall}.
Lemmas \ref{lem:matchExpandIntersectLarge} and \ref{lem:matchExpandIntersectSmall} together show \Cref{lem:matchExpandExpansion}.
All that is left is to prove \Cref{lem:matchExpandGadget}, which we do below.

\matchExpandGadget*
\begin{proof}
Let $X' \cup X$ be a reinforced gadget, with the $h_0$-expander on $X$. Let $S$ be a subset of nodes of size $< |V|/2$. The proof of expansion follows in similar lines as before.

If the intersection with $X$ is large, then we use the matching edges (similar to \Cref{lem:matchExpandBothLarge}). Concretely, if $|S \cap X| > 0.9 |X|$, then since $|X'| \le |X|$ and $|S| < |V|/2$, we get that $|S \cap X'| < (1/9) \cdot |S \cap X|$. Thus
\begin{align*}
|E(S, \bar{S})|
&\ge |S \cap X| - |S \cap X'| \tag*{(by the matching on $X' \cup X$)} \\
&\ge (8/9) \cdot |S \cap X| \tag*{(by constraint on $|S \cap X'|$)} \\
&\ge (8/9) \cdot (9/10) \cdot |X| \tag*{(by constraint on $|S \cap X|$)} \\
&\ge (4/5) \cdot |S| \tag*{(since $|X| \ge |V|/2$ and $|S| \le |V|/2$)}
\end{align*}
If the intersection with $X$ is of medium size, then we use the expander on $X$ (similar to \Cref{lem:matchExpandIntersectMed}). Concretely, if $0.1 |X| \le |S \cap X| \le 0.9 |X|$, then let $T$ be the smaller of the two sets $S \cap X$ and $\bar{S} \cap X$. Then
\begin{align*}
|E(S, \bar{S})|
&\ge |E(S \cap X, \bar{S} \cap X)| \\
&= |E(T, X \setminus T)| \tag*{(by definition of $T$)} \\
&\ge h_0 \cdot |T| \tag*{(by expansion)}\\
&\ge h_0 \cdot 0.1 \cdot |X| \tag*{(by constraint on $|S \cap X|$)}\\
&\ge h_0 \cdot 0.1 \cdot |S| \tag*{(since $|S| \le |X|$)}
\end{align*}
We are left with the case when the intersection with $X$ is small, namely, $|S \cap X| < 0.1 |X|$. If $|S \cap X'| > 2 |S \cap X|$, then we use the matching edges again.
\begin{align*}
|E(S, \bar{S})|
&\ge |S \cap X'| - |S \cap X| \tag*{(by the matching on $X' \cup X$)}\\
&= (1/3) \cdot (2 |S \cap X'| - 4 |S \cap X| + |S \cap X'| + |S \cap X|) \\
&\ge (1/3) \cdot (|S \cap X'| + |S \cap X|) \tag*{(by constraint on $|S \cap X'|$)}
\end{align*}
which leaves us with the final case when $|S \cap X'| \le 2 |S \cap X|$. Here, we use the expander on $X$
\begin{align*}
|E(S, \bar{S})|
&\ge |E(S \cap X, \bar{S} \cap X)| \\
&\ge h_0 \cdot |S \cap X| \tag*{(by expansion)} \\
&\ge (1/3) \cdot h_0 \cdot (|S \cap X| + | S \cap X'|) \tag*{(by constraint on $|S \cap X'|$)}
\end{align*}
which shows that the reinforced gadget has constant expansion.
\end{proof}

\subsection{Power-law Graph}
\label{ssec:match_powerlaw}

We first make the reduction graph robust with respect to degree changes. We use the following static graph in our reduction.
\begin{itemize}
	\item A reduction gadget with one subgadget of size $2n+1$, on a set $L_1 \cup L_2$ as earlier.

	\item A reduction gadget with $2n$ subgadgets of size $2n+2$, on a set $L_3 \cup L_4$. The subgadgets are labelled $LG[i]$ for $1 \le i \le 2n$, and the nodes of subgadget $LG[i]$ are labelled $L_3[i,j]$ or $L_4[i,j]$ for $0 \le j \le n$ depending on whether the node is in $L_3$ or $L_4$. The path in each subgadget goes from $L_3[i,0]$ to $L_4[i,n]$.

    \item A copy of the above structure, with node sets marked $R_i$ instead of $L_i$.

	\item If $M_{ij} = 1$, then add the edges $(L_4[i,j], R_4[j,i])$ and $(L_4[n+i, n+j], R_4[n+j, n+i])$.

	\item If $M_{ij} = 0$, then add the edges $(L_4[i,j], R_4[n+j,n+i])$ and $(L_4[n+i, n+j], R_4[j, i])$.

	\item The edges for an input pair of vectors $(u,v)$ will be detailed later.
\end{itemize}

\begin{table}[t]
    \small
    \renewcommand{\arraystretch}{1.2}
    \centering
    \begin{tabular}{|c|c|c|c|c|}
    \hline
    Layer & Deg $1$ & Deg $2$ & Deg $3$ \\
    \hline
    $L_1$ & $0$ & $n$ & $0$ \\
    \hline
    $L_2$ & $1$ & $n$ & $0$ \\
    \hline
    $L_3$ & $2n$ & $2n^2$ & $0$ \\
    \hline
    $L_4$ & $0$ & $2n$ & $2n^2$ \\
    \hline
    \end{tabular}
    \caption{Degree distribution of the nodes on the left side of the reduction graph.}
    \label{tab:match_powerlaw_degrees}
\end{table}

The degree distribution for each layer on the left side of the reduction gadget in the current instance, before adding any edges for $(u,v)$, is given in~\Cref{tab:match_powerlaw_degrees}. For ease of notation, let us use $(d, N_d)$ to denote that there are $N_d$ nodes of degree $d$ in the graph. Thus the degree distribution in the entire reduction graph is as follows: $(1, 4n+2), (2, 4n^2+8n), (3, 4n^2)$.
Some of the nodes will change their degree when we add edges for the input vector pair $(u,v)$, and we take care of these changes later.

Let $\beta > 2$ be the exponent for which we want to show our lower bound, and $N$ be the number of nodes we need in our reduction graph. First, choose $N$ such that \[
N > \zeta(\beta) \cdot \max \{ (2N_1 + 2n), (2N_2+2n) \cdot 2^{\beta}, (2N_3+2n) \cdot 3^{\beta} \}
,\] and pick any power-law graph $G$ on $N$ nodes. If $N'_d$ is the number of nodes of degree $d$ in $G$, then by construction we get that $N'_d > 2N_d + 2n$.

We would like to essentially embed our reduction graph into this power-law graph. For this, we would like to reduce $N'_d$ by exactly $N_d$ for $d \in \{1, 2, 3\}$, which would allow us to embed our graph into $G$. In what follows, we use the fact that $N_1 < N_3$. We ``make space'' for our nodes as follows.
\begin{itemize}
    \item Do the following $N_1$ times: Pick a node $u \in G$ of degree $1$. Let $w$ be its neighbour. Since $\deg(w) < \sqrt{N}$ and there are $\ge N'_3 - N_3 > n^2$ nodes of degree $3$ in $G$, there exists a node $v \in G$ of degree $3$ such that $v$ has a neighbour $x \in N(v)$ with $(x, w) \not\in G$.
    Delete the edges $(u, w), (v, x)$ and add the edge $(x, w)$.
    This gives us a node $u$ which we can assign degree $1$ to, in our reduction graph. We have also converted a degree $3$ node to a degree $2$ node, which we take care of in the next step.
    \item Do the following $N_2 + N_1$ times: Pick a node $u \in G$ of degree $2$. Let $u_1, u_2$ be its neighbours. Since $\deg(u_1) + \deg(u_2) < 2\sqrt{N} $, and there are $\ge N'_2 - N_2$ nodes of degree $2$ left in $G$, there exists a node $v \in G$ of degree $2$ that has neighbours $v_1, v_2$, with $u_i \neq v_j$. Remove the edges $(u, u_i), (v, v_i)$, and add the edge $(u_i, v_i)$. This gives us two nodes $u, v$ which we can assign degree $1$ to, in our reduction graph.
    \item Similarly, free up $N_3 - N_1$ nodes of degree $3$ for our reduction graph.
\end{itemize}
Note that we will have at most one extra node of each degree which we can leave unused because of the slack allowed in approximate power-law graphs.

We can now embed our reduction graph into a power-law graph with at most two nodes having different degrees, since there are an even number of nodes of degree $2$ in our reduction graph by parity, and we requisitioned degree $1$ nodes one-by-one and not in pairs. Thus, in particular, we are at most one degree $2$ node and one degree $3$ node away from a perfect power law graph.

We now come to the question of the input vectors $(u, v)$. If $u_i = 1$, then we add the edge $(L_2[i], L_3[i,0])$. Note that this changes the degree distribution in the following way: One degree $2$ node increases to degree $3$, and a degree $1$ node increases to degree $2$. We need to adjust for this in the power law graph, while making sure that the size of the maximum matching in the remaining graph is still known. We do this as follows:

During preprocessing, pick $2n$ disjoint pairs of nodes $(a_i, b_i)$ in the graph $G$ such that $\deg(a_i) = 2$ and $\deg(b_i) = 3$, such that $\exists\, c_i \in N(a_i), d_i \in N(b_i)$ with $(c_i, d_i) \not\in G$. $G_0$ is the graph $G$. Let $G_j$ be the graph with the edges $(a_i, c_i), (b_i, d_i)$ deleted and the edges $(c_i, d_i)$ added for all $1 \le i \le j$. Since $\poly(n)$ preprocessing is allowed in the \omv{} conjecture, we can afford to find the sizes of the maximum matchings in all the graphs $G_j$ before we receive any input. Denote the sizes of these matchings as $m_j$, $0 \le j \le 2n$.

On input $(u, v)$, we do the following:
\begin{itemize}
    \item If $u_i = 1$, then add the edge $(L_2[i], L_3[i,0])$.
    \item If $v_i = 1$, then add the edge $(R_2[i], R_3[i,0])$.
    \item Let $k = \supp(u) + \supp(v)$. Delete the edges $(a_i, c_i), (b_i, d_i)$ and add the edges $(c_i, d_i)$ for all $1 \le i \le k$.
\end{itemize}

Now we ask for the size of the maximum matching in this graph. Let $\overline{G}$ be the subgraph of $G$ which is the reduction graph, and let $\overline{N}$ be the number of nodes in $\overline{G}$. Note that we already know the size of the maximum matching in $G \setminus \overline{G}$ to be $m_k$. $uMv = 1$ if and only if a maximum matching restricted to $\overline{G}$ is perfect, and since $\overline{G}$ is disjoint from $G \setminus \overline{G}$, if and only if the maximum matching on $G$ is of size $m_k + \frac{\overline{N}}{2}$. We then roll back the graph to its previous state and process the new input vector pair.

We make $O(n)$ updates and $1$ queries for each input pair, and the graph consists of $\Theta(n^2)$ nodes, which gives us the same lower bounds as in the constant-degree reduction.

\section{\texorpdfstring{Lower Bounds for Dynamic $(s,t)$-Shortest Path}{Lower Bounds for Dynamic (s,t)-Shortest Path}}
\label{sec:full_st_LBs}

In this section, we present our lower bound results for the dynamic $(s,t)$-shortest path problem.
In \Cref{ssec:st_const_degree}, we give a lower bound for dynamic $(s,t)$-distance on graphs with maximum degree $3$. We extend this lower bound to $(3-\delta)$-approximations in \Cref{ssec:st_const_apx}. In \Cref{ssec:st_vary_degree}, we show that the distinction between the unbounded and constant-degree reductions is not discrete, by giving a lower bound reduction parameterized on the maximum degree allowed in the graph. In \Cref{ssec:st_expand}, we show that the lower bound on constant-degree graphs holds even on expanders, by constructing a more involved reduction graph. Finally, we prove the power-law graph lower bounds in \Cref{ssec:st_powerlaw}.

\subsection{Constant-Degree Graph}
\label{ssec:st_const_degree}

Consider the \oumv{} problem on vectors of length $n$ and an $n \times n$ matrix. We first perform a simple reduction that shows that maintaining $(s,t)$-distance is hard even on graphs where the maximum degree is $3$.

\begin{restatable}{theorem}{stConst}
	\label{thm:stConst}
	For any constant $\epsilon>0$,
	there is no dynamic algorithm maintaining $(s,t)$-distance, SSSP or APSP,
	on all $N$-node bipartite graphs with maximum degree $\Delta \le 3$,
	with amortized \ub{1/2} update time and \ub{1} query time,
	unless the \omv{} conjecture is false.
\end{restatable}

Since the original reduction~\cite{HenzingerKNS15} could possibly have unbounded degree, we use \emph{binary forests} to sparsify our reduction graph.

\begin{definition}[Binary Forest]
A \emph{binary forest} of $x$ trees of height $y$ is a graph composed of $x$ disjoint binary trees, each of height $y$.
\end{definition}

We naturally split the nodes of a binary forest between \emph{internal nodes} and \emph{leaves}. Intuitively, we would like to replace a high degree node of the original reduction with a binary forest to moderate the maximum allowed degree.
Note that a binary forest has $x\cdot2^y$ leaves in total.

\subsubsection{Static Graph}
\label{sssec:st_static}

We use the following static graph as the base for our reduction:
\begin{itemize}
	\item A $(\log n)$-depth binary forest with a single tree. The set of $n-1$ internal nodes is marked $L_1$ ($L$ for left) and the $n$ leaves $L_2$; the root of the tree is the source node $s$, and the $n$ nodes of $L_2$ are marked as $L_2[i]$ for $1\le i \le n$.

	\item A $(\log n)$-depth binary forest with $n$ trees. The $n(n-1)$ internal nodes are marked $L_3$ and the $n^2$ leaves $L_4$. The roots of each of the $n$ trees are marked as $L_3[i]$, for $1 \le i \le n$, and the leaves of the tree with root $L_3[i]$ are marked $L_4[i,j]$, for $1 \le j \le n$.

	\item A copy of the above structure, with node sets marked $R_1,R_2,R_3,R_4$ instead of $L_1,L_2,L_3,L_4$, respectively. The root of the single tree of $R_1$ is the target node $t$.

	\item Edges from $L_4$ to $R_4$ by the matrix $M$, as detailed next.

	\item For an input pair of vectors $(u,v)$, edges between $L_2$ and $L_3$ by $u$, and between $R_2$ and $R_3$ by $v$, as detailed next.
\end{itemize}

The total number of nodes in the reduction graph is $N = 4n^2 + 2n - 2 = \Theta(n^2)$.

\subsubsection{Input-Dependent Edges}
\label{sssec:st_input}

We add the following edges depending on the input matrix $M$ and vectors $u, v$:

\begin{itemize}
    \item For the matrix $M$, add the edge $(L_4[i,j],R_4[j,i])$ if $M_{ij}=1$.
    \item Given an input vector $u$, for each $i\in[n]$, add the edge $(L_2[i], L_3[i])$ if $u_i=1$.
    \item Given an input vector $v$, for each $j\in[n]$, add the edge $(R_2[j], R_3[j])$ if $v_j=1$.
\end{itemize}

\subsubsection{Distances in the Graph}
\label{sssec:st_distances}

We now show the correctness of the reduction, by considering the $(s,t)$ distance in different scenarios.

\begin{restatable}{lemma}{stDistLem}[constant-degree reduction]
\label{lem:stDistLem}
$uMv=1$ if and only if $\dist(s,t) \le 4\log n+3$. Moreover, the graph is bipartite.
\end{restatable}
\begin{proof}
We first show that any path from $s$ to $t$ has to be of length at least $4 \log n + 3$, by partitioning the node set into layers. We maintain the property that a node at level $\ell$ can have neighbours only in levels $\ell -1$, $\ell$, or $\ell + 1$. The layering is as follows: The layer of a node in $L_1 \cup L_2$ is its distance from $s$; in $L_3 \cup L_4$ is its distance from its closest root, plus $\log n + 1$; in $R_3 \cup R_4$ is its distance from its closest leaf, plus $2\log n + 2$; in $R_1 \cup R_2$ is its distance from its closest leaf, plus $3 \log n + 3$. Specifically, the layer of node $t$ is $4 \log n + 3$. Thus $\dist(s,t) \ge 4\log n + 3$, regardless of $u$, $M$, and $v$.
Moreover, the the fact that edges only connects consecutive layers implies that the graph is bipartite.

($\implies$) If $uMv=1$, then there exists indices $i, j$ such that $u_i = M_{ij} = v_j = 1$. Then consider the path $P$ composed of the following sub-paths:
\begin{itemize}
    \item $P_1$ is the shortest path from $s$ to $L_3[i]$, which follows the tree $L_1 \cup L_2$ and then the $(L_2[i], L_3[i])$ edge ($\log n + 1$ edges).
    \item $P_2$ is the shortest path from $L_3[i]$ to $R_4[j,i]$, which follows the $i^{th}$ tree rooted at $L_3[i]$ and then the $(L_4[i,j], R_4[j,i])$ edge ($\log n + 1$ edges).
    \item $P_3$ is the shortest path from $R_4[j,i]$ to $R_2[j]$, which follows the $j^{th}$ tree rooted at $R_3[j]$ and then the $(R_3[j], R_2[j])$ edge ($\log n + 1$ edges).
    \item $P_4$ is the shortest path from $R_2[j]$ to $t$, which follows the tree $R_1 \cup R_2$ ($\log n$ edges).
\end{itemize}
$P$ is then a path of length $4 \log n + 3$ from $s$ to $t$.

($\impliedby$) Assume that there is a path of length $4\log n+3$ from $s$ to $t$. By the layering, each edge in the path must connect a layer $\ell$ node and a layer $(\ell+1)$ node, and there is exactly one such edge in the path for each $\ell\in[4\log n+2]$. Next, we use these two facts (sometimes implicitly) to show that the path must have the form of the above described path, and conclude that $uMv=1$.

The path must contain an edge from $L_2$ to $L_3$. The only such edges are of the form $(L_2[i], L_3[i])$ for some $i\in [n]$, implying $u_i=1$. From there, the path must continue to some leaf $LU[i,j]$ of the tree rooted at $LU[i]$. Since the only edge from $L_4[i,j]$ that strictly increases in level is the edge $(L_4[i,j], R_4[i,j])$, the path must contain such an edge, implying $M_{ij}=1$ for some $j \in [n]$. From $R_4[j,i]$, the path must continue to the root $R_3[j]$, and then to $R_2[j]$ over an edge $(R_3[j],R_2[j])$, implying $v_j=1$. The path ends trivially by going from $R_2[j]$ to the root $t$. Thus $uMv=1$.
\end{proof}

\subsubsection{Complexity of the Reduction}
\label{sssec:st_complexity}
We are now ready to prove the theorem.

\stConst*
\begin{proof}
Consider the reduction graph above, which is bipartite by \cref{lem:stDistLem}.
It consists of $N = 4n^2 + 2n -2 = \Theta(n^2)$ nodes. Every time we get a new $(u,v)$ input vector pair, we delete all the edges between $L_2 \times L_3$ and $R_2 \times R_3$ and insert edges according to the new input vectors. This takes $O(n)$ updates in total. After that, we query once for the $(s,t)$-distance in this new graph, and return $1$ if and only if $\dist(s,t) = 4\log n + 3$.

Thus for each pair of input vectors, we perform $O(n)$ updates and $O(1)$ query. In total, checking $n$ pairs takes us $O(n^2)$ updates and $O(n)$ query. If there were an algorithm for $(s,t)$-distance on constant-degree graphs with update time \ub{1/2} (i.e., $O(n^{1-2\epsilon})$) and query time \ub{1} (i.e., $O(n^{2-2\epsilon})$), then we can decide if $uMv=1$ for all $n$ pairs in $O(n^{3-2\epsilon})$ time, contradicting the \omv{} conjecture.
\end{proof}

\subsection{\texorpdfstring{$(3-\delta)$-Approximation Lower Bound for Constant-Degree Graphs}{(3-d)-Approximation Lower Bound for Constant-Degree Graphs}}
\label{ssec:st_const_apx}

We show that the above lower bounds on $(s,t)$-distances also holds for $(3-\delta)$-approx $(s,t)$-distances by minimally modifying the reduction graph to exploit the following observation: If $uMv = 0$, then every path from $s$ to $t$ in the simple reduction graph needs to take at least three edges corresponding to the matrix $M$, as opposed to just one such edge when $uMv=1$.

We use the same reduction graph as before, but with one important difference: Earlier, we added a single edge $(L_4[i,j], R_4[j,i])$ if $M_{ij} = 1$. Now, we add a path with $\Theta(\log n)$ new nodes between $L_4[i,j]$ and $R_4[j,i]$. Note that since we only do this for $M$ and not for every new input vector pair, we can do this in just polynomial pre-processing time, which is allowed for in the \omv{} conjecture.

Formally, let $\alpha = \left\lceil \frac{12}{\delta} - 4 \right\rceil $.
Add $n^2 \cdot \alpha \log n$ new nodes to the reduction graph, with the nodes labelled $v_{ij}[k]$ for $1 \le k \le \alpha \log n$ and $1 \le i,j \le n$.
If $M_{ij} = 1$, add the path $L_4[i,j], v_{ij}[1], \ldots, v_{ij}[\alpha \log n]$, $R_4[i,j]$ to the reduction graph, while otherwise all the nodes $v_{ij}[k]$ remain disconnected.

The following lemma is an analogue of \Cref{lem:stDistLem}, and the proof follows from the proof of \Cref{lem:stDistLem} and the above observation.

\begin{restatable}{lemma}{stDistApxLem}
\label{lem:stDistApxLem}
If $uMv=1$, then $\dist(s,t) = (4 + \alpha)\log n+2$, and otherwise $\dist(s,t) \ge (4 + 3\alpha)\log n +2$.
Moreover, the graph is bipartite.
\end{restatable}
\begin{restatable}{corollary}{stConstApx}
\label{cor:stConstApx}
	For any constant $\epsilon>0$,
	there is no dynamic algorithm maintaining $(3-\delta)$-approximate $(s,t)$-distance, SSSP or APSP,
	on all $N$-node bipartite graphs with maximum degree $\Delta \le 3$,
	with amortized \ub{1/2} update time and \ub{1} query time,
	unless the \omv{} conjecture is false.
\end{restatable}
\begin{proof}
From \Cref{lem:stDistApxLem}, we get that any approximate reported distance $d$ where $(4+\alpha) \log n + 2 \le d < (4+3\alpha) \log n + 2$ would imply that $uMv=1$. Note that
\begin{align}
\frac{(4 + 3 \alpha) \log n + 2}{(4+\alpha) \log n + 2}
&= 3 - \frac{8 \log n + 4}{(4 + \alpha) \log n + 2} \nonumber \\
&\ge 3 - \frac{12\log n}{(4+\alpha)\log n} \label[ineq]{eq:triv_apx} \\
&\ge 3 - \delta \label[ineq]{eq:alpha_apx}
\end{align}
where \Cref{eq:triv_apx} follows from the fact that $4\log n \ge 4$, and \Cref{eq:alpha_apx} follows from the definition of $\alpha$. Thus any $(3-\delta)$-approx algorithm would be able to distinguish between $uMv=1$ and $uMv=0$.

The number of nodes in the reduction graph is $N = \Theta(n^2\log n)=O(n^{2+\epsilon})$. Thus for each pair of input vectors, we perform $O(n)$ updates and $O(1)$ queries.
In total, checking $n$ pairs takes us $O(n^2)$ updates and $O(n)$ queries.
If there were an algorithm for $(3-\delta)$-approx $(s,t)$-distance on constant-degree graphs with update time $\ub{1/2}=O(n^{1 - 3\epsilon/2 - \epsilon^2)}$ and query time $\ub{1}=O(n^{2 - \epsilon - \epsilon^2})$, then
we can decide if $uMv=1$ for all $n$ pairs in $O(n^{3-c})$ time for some constant $c$, contradicting the \omv{} conjecture.
\end{proof}

\subsection{Varying-Degree Graph}
\label{ssec:st_vary_degree}

We present a reduction that gives a lower bound parameterized on the maximum degree in the graph.

\begin{restatable}{theorem}{stVary}
	\label{thm:stVary}
	For any $0\leq t\leq 1$ and any constant $\epsilon>0$,
	there is no dynamic algorithm maintaining $(s,t)$-distance, SSSP or APSP,
	on all $N$-node bipartite graphs with maximum degree $\Delta = O(N^t)$,
	with amortized \ub{\frac{1+t}{2}} update time and \ub{1+t} query time,
	unless the \omv{} conjecture is false.
\end{restatable}

\paragraph*{The reduction graph}
\label{sssec:st_vary_graph}
For a value $0 \le t \le 1$ (that may depend on $N$), we construct a reduction graph on $N$ nodes that has maximum degree $\Theta(N^t)$. The node sets $L_1, L_2, R_1, R_2$ and the edges that depend on $u$ and $v$ are the same as in the previous construction in~\Cref{ssec:st_const_degree}.
We detail the changes for $L_3, L_4, R_3, R_4$ and $M$ below.

\begin{itemize}

	\item A $\left(\frac{1-t}{1+t} \cdot \log n\right)$-depth binary forest with $n$ trees. The $n\cdot \left(n^{\frac{1-t}{1+t}}-1\right)$ internal nodes are marked $L_3$ and the $n \cdot n^{\frac{1-t}{1+t}}$ leaves are marked $L_4$.
	The roots of each of the $n$ trees are marked as $L_3[i]$, for $1 \le i \le n$, and the leaves of the tree with root $L_3[i]$ are marked $L_4[i,j]$, for $1 \le j \le n^{\frac{1-t}{1+t}}$, and similarly for $R_3$ and $R_4$.

    \item For the matrix $M$, if $M_{ij}=1$, let $i' = \left\lceil i \cdot n^{-2t/(t+1)} \right\rceil$ and $j' = \left\lceil j \cdot n^{-2t/(t+1)} \right\rceil $. Add the edge $(L_4[i,j'],R_4[i',j])$ to the graph.
\end{itemize}

Note that the distances in this reduction graph are the same as in the constant-degree reduction graph by a similar proof as used to prove~\Cref{lem:stDistLem}.
Furthermore, each node in $L_4$ is connected to at most $n^{2t/(t+1)}$ nodes in $R_4$, and each node in $R_4$ is connected to at most $n^{2t/(t+1)}$ nodes in $L_4$.

We can now prove \Cref{thm:stVary}.
\begin{proof}
The number of nodes in the reduction graph above is dominated by the number of nodes in $L_4$ and $R_4$.
Thus, the total number of nodes in the reduction graph is $N = \Theta(n^{2/(t+1)})$ nodes. Since each node not in $L_4$ or $R_4$ have at most $3$ edges adjacent on it, its degree is  trivially $O(N^t)$.
Every node in $L_4$ and $R_4$ has one tree edge incident on it, and at most $n^{2t/(t+1)}$ edges of $M$ incident on it by construction. Thus the maximum degree in the graph is $O(n^{2t/(t+1)}) = O(N^t)$ as claimed. The rest is similar to \Cref{thm:stConst}. Every time we get a new $(u,v)$ input vector pair, we delete all the edges between $L_2 \times L_3$ and $R_2 \times R_3$ and insert edges according to the new input vectors.

For each pair of input vectors, we thus perform $O(n)$ updates and $O(1)$ queries.
In total, checking $n$ pairs takes $O(n^2)$ updates and $O(n)$ queries.
If there was an algorithm for $(s,t)$-distance on graphs with maximum degree bounded by $N^t$ with update time \ub{\frac{1+t}{2}} (i.e., $O(n^{1-2\epsilon})$) and query time \ub{1+t} (i.e., $O(n^{2-2\epsilon})$), then we can decide if $uMv=1$ for all $n$ pairs in $O(n^{3-2\epsilon})$ time, contradicting the \omv{} conjecture.
\end{proof}

\subsection{Expander Graph}
\label{ssec:st_expand}

For our expander lower bound, we use the following gadgets, called \emph{reinforced forests}.
\vspace{2ex}

\noindent
\begin{minipage}{.8\textwidth}
\begin{definition}[Reinforced forest]
A \emph{reinforced forest} of $x$ trees of height $y$ is a graph composed of
$x$ disjoint binary trees, each of them of height $y>0$.
These trees have $x2^y$ leaves in total;
consider a degree-$d$ expander graph on $x2^y$ nodes, choose an arbitrary bijection between the expander nodes and the forest's leaves, and add the expander edges to the leaves accordingly.
Finally, in order to guarantee that the graph is bipartite, on each edge added from the expander, we add a dummy node.
The resulting graph is the reinforced forest.
\end{definition}
\end{minipage}
\begin{minipage}{.2\textwidth}
\begin{center}
\vspace*{2ex}
	\includegraphics[scale=0.7]{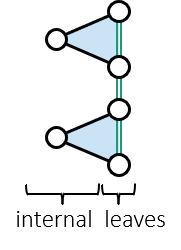}
\end{center}
\end{minipage}
\vspace{1ex}

We naturally split the nodes of a reinforced forest between \emph{internal nodes} and \emph{leaves}.

\paragraph*{Reduction Graph}
\label{sssec:st_expand_static}
\begin{figure}
	\centering
	\includegraphics[scale=0.7]{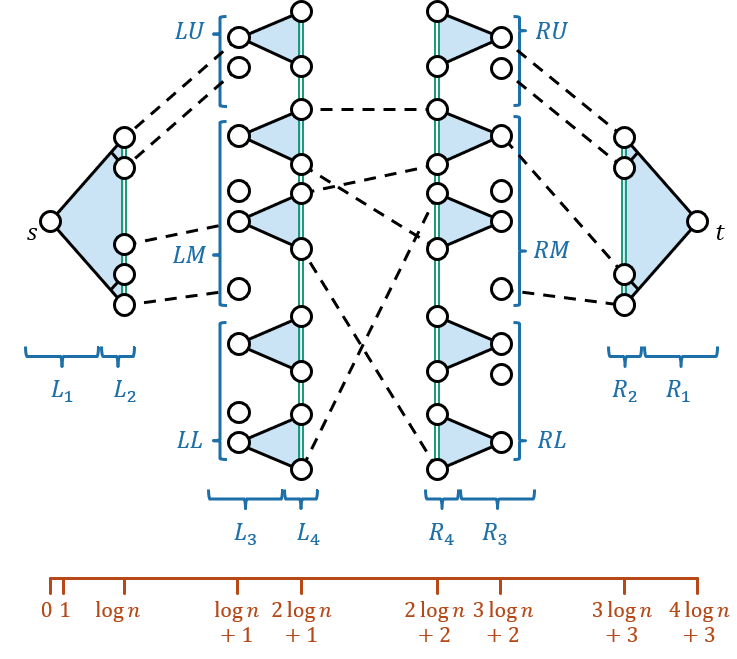}
	\caption{The reduction graph.
		The input-dependent edges are dotted.}
	\label{fig:new_st_expand_graph}
\end{figure}

We use the following graph as the base for our reduction (see \cref{fig:new_st_expand_graph}).
\begin{itemize}
	\item A $(\log n)$-depth reinforced forest with a single tree. The set of $n-1$ internal nodes is marked $L_1$ and the $n$ leaves $L_2$; the root of the tree is the source node $s$, and the $n$ nodes of $L_2$ are marked as $L_2[i]$ for $1\le i \le n$.

	\item A $(\log n)$-depth reinforced forest with $3n$ trees. The $3n(n-1)$ internal nodes are marked $L_3$ and the $3n^2$ leaves $L_4$. We split this reinforced forest into sets of $n$ trees each, and label them $LU, LM, LL$ (upper, middle, and lower). The roots of each of the $n$ trees are marked as $LX[i]$, for $1 \le i \le n$, $X \in \{U,M,L\}$, and the leaves of the tree with root $LX[i]$ are marked $LX[i,j]$, for $1 \le j \le n$.

	\item A copy of the above structure, with node sets marked $R_i$ instead of $L_i$, respectively. The root of the single tree of $R_1$ is the target node $t$.

    \item For the matrix $M$, add the edge $(LM[i,j],RM[j,i])$ if $M_{ij}=1$, and add the edges $(LM[i,j], RL[j,i])$ and $(LL[i,j], RM[j,i])$ otherwise.

    \item Given an input vector $u$, for each $i\in[n]$, add the edge $(L_2[i], LM[i])$ if $u_i=1$, and $(L_2[i], LU[i])$ otherwise.

    \item Given an input vector $v$, for each $j\in[n]$, add the edge $(R_2[j], RM[j])$ if $v_j=1$, and $(R_2[j], RU[j])$ otherwise.
\end{itemize}

The following lemma characterizes $(s,t)$-distance in the expander reduction graph based on the value of $uMv$.
We split the graph nodes into layers (indicated at the bottom of \cref{fig:new_st_expand_graph}) and use the fact that edges can only connect consecutive layers in order to prove this lemma in~\Cref{sssec:st_expand_dist_proof}.

\begin{restatable}{lemma}{stExpandDistLem}
\label{lem:stExpandDistLem}
$uMv=1$ if and only if $\dist(s,t) \le 4\log n+3$. Moreover, the graph is bipartite.
\end{restatable}

The number of nodes in the reduction graph is dominated by the nodes in $L_4 \cup R_4$, yielding the same asymptotic lower bounds as in the constant-degree case (\cref{ssec:st_const_degree}).
When a vector pair arrives, we first add all the potential edges and then remove the unnecessary ones,
in a similar fashion as before, to preserve expansion.
We defer the proof of expansion to \Cref{sssec:st_expand_proof}.
Thus, we get the following theorem for expander graphs.
\begin{restatable}{theorem}{stExpand}
    \label{thm:stExpand}
    For any constant $\epsilon>0$,
    there is no dynamic algorithm maintaining $(s,t)$-distance, SSSP or APSP,
    on all $N$-node constant-degree bipartite graphs with constant expansion,
    with amortized \ub{1/2} update time and \ub{1} query time,
    unless the \omv{} conjecture is false.
\end{restatable}

\subsubsection{Distances in the expander graph}
\label{sssec:st_expand_dist_proof}
We prove the following lemma for the expander reduction graph described in \cref{ssec:st_expand}.
This is basically the graph described in \cref{ssec:st_const_degree}, with expander edges added in carefully chosen places, and one dummy node added on each such edge.
Note that adding the dummy nodes does not change the expansion by more than a constant factor,
does not not affect the asymptotic number of nodes,
and the shortest paths never use the expander edges and thus the also stay unaffected by the dummy nodes addition.
To simplify the writing, we thus ignore the dummy nodes in the following.
\stExpandDistLem*
\begin{proof}
The proof uses the same ideas and layering as in the proof of \cref{lem:stDistLem}.
We include the dummy nodes of $L_2$ and $L_4$ in one layer before their (non-dummy) neighbors, and the other dummy nodes in one layer after their neighbors.
This already shows that the graph is bipartite.

($\implies$) This direction of the proof is the same as in the proof of \cref{lem:stDistLem}, since we only add edges, and the same path as earlier still exists in the graph.

($\impliedby$) We still make use of the layering argument for this direction, but the argument is slightly more nuanced than the one for the constant-degree case.
Assume that there is a path of length $4\log n+3$ from $s$ to $t$. By the layering, each edge in the path must connect a layer $\ell$ node and a layer $(\ell+1)$ node, and there is exactly one such edge in the path for each $\ell\in[4\log n+2]$. Next, we use these two facts (sometimes implicitly) to show that the path must have the form of the above described path, and conclude that $uMv=1$.
\begin{itemize}
    \item The path must contain an edge from $L_2$ to $L_3$. Suppose the edge was to a node in $LU$. Since no node in $LU$ has an edge to $R_4$, the path does not always connect two nodes of strictly increasing layers, contradicting the claimed length. Thus the edge is of the form $(L_2[i], LM[i])$ for some $i \in [n]$, implying that $u_i = 1$.

    \item From there, the path must continue to some leaf $LM[i,j]$ of the tree rooted at $LM[i]$. Suppose the path continues to $RL$ instead of $RM$. Since there are no edges from $RL$ to $R_2$, the path would have to connect two  of the non-increasing at least once, which cannot happen on a $4\log n + 3$ length path. Thus the path uses the edge $(LM[i,j], RM[j,i])$, implying that $M_{ij} = 1$.

    \item From $RM[j,i]$, the path must continue to the root $RM[j]$, and then to $R_2[j]$ by using the edge $(RM[j],R_2[j])$, implying $v_j=1$.
\end{itemize}
We have established that there exist indices $i,j\in [n]$ such that $u_i = M_{ij} = v_j= 1$, implying $uMv=1$.
\end{proof}

\subsubsection{Expansion of the expander graph}
\label{sssec:st_expand_proof}

Let us now verify that the graph is indeed an $h_1$-expander graph, for some constant $h_1>0$.
Recall that the reinforced forests contain $h_0$ expanders, for some constant $h_0>0$.

While a similar proof as in the maximum matching expander reduction works for this setting as well, we take a different approach here.
Consider a reinforced forest with a set $U$ of internal nodes, a set $U'$ of leaves, and an $h_0$-expander on the leaves.
Let $S$ be a set of nodes in the forest (we do not bound $|S|$).

\begin{observation}
	\label{obs:strein_medium}
	Let $0<c<1/2$, that may depend on $n$.
	If $c|U'|\leq|S\cap U'|\leq (1-c)|U'|$, then
	$E(S,\bar S)|\geq ch_0|U'|$.
\end{observation}
\begin{claimproof}
	Consider first the case $|S\cap U'|\leq |U'|/2$, in which the expander graph edges on $U'$ guarantee $|E(S,\bar S)\cap U'|\geq h_0|S|$, implying
	$|E(S,\bar S)|\geq|E(S,\bar S)\cap U'|\geq h_0|S|\geq ch_0|U'|$.
	Otherwise, apply the analogous argument on the set $U'\setminus S$.
\end{claimproof}

\begin{observation}
	\label{obs:strein_moreint}
	$|E(S,\bar S)|\geq |U\cap S|-|U'\cap S|$.
\end{observation}
\begin{claimproof}
	Every node in $U\cap S$ is a tree node that has two children.
	In total, these node has $2|U\cap S|$ edges going to children, of which at most $|U\cap S|+|U'\cap S|$ children are in $S$.
	Hence, $|E(S,\bar S)|\geq 2|U\cap S|-(|U\cap S|+|U'\cap S|)=|U\cap S|-|U'\cap S|$.
\end{claimproof}

\begin{observation}
	\label{obs:strein_lessint}
	Let $0<c<1/2$ constant.
	If $|U\cap S|\leq c|U'\cap S|$
	then $|E(S,\bar S)|\geq (0.5-c)|U'\cap S|$.
\end{observation}
\begin{claimproof}
	If $|U\cap S|\leq c|U'\cap S|$, consider the parent nodes of the leaves in $U'\cap S$: these have at least $0.5|U'\cap S|$ distinct parents, of which at most $c|U'\cap S|$ are in $U\cap S$, and the others are in
	$U\cap \bar S$. Hence
	$|E(U\cap \bar S,U'\cap S)|\geq (0.5-c)|U'\cap S|$.
\end{claimproof}

The proof of the following lemma is similar to that of~\Cref{lem:matchExpandGadget}.

\begin{lemma}
	\label{lem: rein-forest is an expander}
	A reinforced forest is an expander.
\end{lemma}
\begin{proof}
	Fix a set $S$ of nodes in the forest, $|S|\leq(|U\cup U'|)/2$.
	We consider different ranges of $|U'\cap S|$, and show that for each of them, $|E(S,\bar S)|\geq h_1|S|$ for some constant $h_1$.
	\begin{enumerate}
		\item
		If $|U'\cap S|>0.9|U'|$, the inequalities $|U|<|U'|$ and $|S|\leq|U\cup U'|/2$ imply $|U\cap S|<0.1|U'|$, and hence
		$|U'\cap S|>(1/9)|U\cap S|>0.1|U\cap S|$.
		Observation~\ref{obs:strein_lessint} implies $|E(S,\bar S)|\geq 0.4|U'\cap S|\geq 0.36|U'|> 0.18|U\cup U'|\geq 0.36|S|$.

		\item
		If $0.1|U'|\leq |U'\cap S|\leq 0.9|U'|$, then by Observation~\ref{obs:strein_medium} we have
		$|E(S,\bar S)|\geq 0.1h_0|U'|$.
		Since $|U|<|U'|$ and $|S|\leq|U\cup U'|/2$, the above implies
		$|E(S,\bar S)|\geq 0.05h_0|U\cup U'|\geq 0.1 h_0|S|$.

		\item
		If $|U'\cap S|<0.1|U'|$, then
		\begin{enumerate}
			\item
			If $|U\cap S|>2|U'\cap S|$,
			add
			$0.5|U\cap S|-1.5|U'\cap S|$ to both sides of the inequality and multiply by $2/3$ to get
			$|U\cap S|-|U'\cap S|>1/3(|U\cap S|+|U'\cap S|)=|S|/3$.
			By Observation~\ref{obs:strein_moreint},
			$|E(S,\bar S)|\geq |U\cap S|-|U'\cap S|>|S|/3$.

			\item
			If $|U\cap S|\leq 2|U'\cap S|$
			then
			$|S|=|U\cap S|+|U'\cap S|\leq 3|U'\cap S|$.
			The expander edges on $U'$ guarantee that
			$|E(S,\bar S)|\geq h_0|U'\cap S|\geq(h_0/3)|S|$.
			\qedhere
		\end{enumerate}
	\end{enumerate}
\end{proof}

\begin{lemma}
The reduction graph is an expander.
\end{lemma}
\begin{proof}
Consider a set $S\subseteq V$ of $|S|\leq N/2=6n^2-n-1$ nodes.
Note that the node sets in the graph have the following sizes:
$|L_1|=|R_1|=n-1$; $|L_2|=|R_2|=n$; $|L_3|=|R_3|=3n^2-3n$; $|L_4|=|R_4|=3n^2$.
We show that there exists a constant $h_1$ such that $|E(S,\bar S)|\geq h_1|S|$,
by considering different ranges for the set size $|L_4\cap S|$ (which ranges in $0,\ldots, 3n^2$).
In most cases, we show that $|E(S,\bar S)|\geq cn^2$ for some constant $c$, which is enough since $|S|<6n^2$.
\begin{enumerate}
	\item
	If $|L_4\cap S|>2.9n^2$, then consider following sub-cases.
	\begin{enumerate}
		\item
		If $|R_4\cap S|>2.1n^2$, then since $|S|\leq 6n^2-n-1$, we have $|L_3\cap S|< n^2$, which implies
		$|L_4\cap S|>2.9|L_3\cap S|$.
		By Observation~\ref{obs:strein_lessint},
		$|E(S,\bar S)|\geq 0.15|L_4\cap S|>0.4n^2$.

		\item
		If $0.1n^2\leq|R_4\cap S|\leq2.1n^2$, the expander edges on $R_4$ (Observation~\ref{obs:strein_medium} with $c=0.03$) guarantee $E(S,\bar S)|\geq 0.03h_0n^2$.

		\item
		If $|R_4\cap S|<0.1n^2$, note that there are $n^2$ edges in $L_4\times R_4$, of which
		at most $0.1n^2$ are in $(L_4\cap\bar S)\times (R_4\cap\bar S)$ (since $L_4\cap S$ is large),
		at most $0.1n^2$ are in $(L_4\cap S)\times (R_4\cap S)$ (since $R_4\cap S$ is small),
		and the remaining edges, at least $0.8n^2$ of them, are in $E(S,\bar S)$.
	\end{enumerate}

	\item
	If $0.1n^2\leq|L_4\cap S|\leq2.9n^2$, then as in case 1(b),
	the expander edges on $L_4$
	(Observation~\ref{obs:strein_medium} with $c=0.03$)
	guarantee $E(S,\bar S)|\geq 0.03h_0n^2$.

	\item
	If $|L_4\cap S|<0.1n^2$, consider the following sub-cases.
	\begin{enumerate}
		\item
		If $|R_4\cap S|>2.1n^2$, we are in a case symmetric to case 1(c):
		of the $n^2$ edges in $L_4\times R_4$,
		at most $0.1n^2$ are in $(L_4\cap\bar S)\times (R_4\cap\bar S)$ (since $L_4\cap S$ is small),
		at most $0.1n^2$ are in $(L_4\cap S)\times (R_4\cap S)$ (since $R_4\cap S$ is large),
		and the remaining edges, at least $0.8n^2$ edges, are in $E(S,\bar S)$.

		\item
		If $0.1n^2\leq|R_4\cap S|\leq(2.1)n^2$, again the expander edges on $R_4$
		(Observation~\ref{obs:strein_medium} with $c=0.03$) guarantee $E(S,\bar S)|\geq 0.03h_0n^2$.

		\item
		In the last case, namely $|R_4\cap S|<0.1n^2$, we are bound to analyze the subgraphs on $L=L_1\cup L_2\cup L_3\cup L_4$ and on $R=R_1\cup R_2\cup R_3\cup R_4$ separately.
		We will show that there is a constant $h_2$ such that
		$|E(L\cap S,L\cap\bar S)|\geq h_2|L\cap S|$ and
		$|E(R\cap S,R\cap\bar S)|\geq h_2|R\cap S|$,
		which implies
		$|E(S,\bar S)|
		\geq |E(L\cap S,L\cap\bar S)|+|E(R\cap S,R\cap\bar S)|
		\geq h_2|L\cap S|+  h_2|R\cap S|=h_2|S|$, as desired.

		We focus on $L$; the proof for $R$ is analogous.
		\begin{enumerate}
			\item
			If $0.1n\leq|L_4\cap S|<0.1n^2$,
			\begin{enumerate}
				\item
				If $|L_3\cap S|\geq 2|L_4\cap S|$, then by Observation~\ref{obs:strein_moreint}, we have
				$|E(L\cap S,L\cap\bar S)|\geq |L_3\cap S|-|L_4\cap S|$, and we
				bound this difference from below several times.
				By $|L_3\cap S|\geq 2|L_4\cap S|$, we have
				$|L_3\cap S|-|L_4\cap S|\geq
				0.5|L_3\cap S|$.
				By the same inequality
				$|L_3\cap S|-|L_4\cap S|\geq|L_4\cap S|$.
				In addition,
				$|L_4\cap S|\geq 0.1n\geq 0.05|L_1\cup L_2|$.
				Hence, $|E(L\cap S,L\cap\bar S)|\geq
				(0.5|L_3\cap S|+|L_4\cap S|+0.05|L_1\cup L_2|)/3\geq
				0.01|L\cap S|$,
				and we are done.

				\item
				If $|L_3\cap S|\leq 2|L_4\cap S|$, then
				$|(L_3\cup L_4)\cap S|\leq |(L_3\cup L_4)|/2$,
				and by Lemma~\ref{lem: rein-forest is an expander} we have
				$|E(L\cap S,L\cap\bar S)|\geq
				h_1|(L_3\cup L_4)\cap S|$.
				We also have
				$|(L_3\cup L_4)\cap S|
				\geq|(L_4\cap S|\geq0.1n\geq 0.05|L_1\cup L_2|$.
				Hence, $|E(L\cap S,L\cap\bar S)|\geq 0.02h_1|L\cap S|$
				and we are done.
			\end{enumerate}

			\item
			If $|L_4\cap S|<0.1n$,
			\begin{enumerate}
				\item
				If $|L_3\cap S|\geq 0.2n$, then again by Observation~\ref{obs:strein_moreint}, we have
				$|E(L\cap S,L\cap\bar S)|\geq |L_3\cap S|-|L_4\cap S|
				\geq |L_3\cap S|/2$.

				Since $|L_3\cap S|/2 > |L_4\cap S|/4$ and also
				$|L_3\cap S|/2 \geq 0.1n>0.05|L_1\cup L_2|$, we have
				$|E(L\cap S,L\cap\bar S)|\geq0.01|L_3\cap S|$.

				\item
				If $|L_3\cap S|< 0.2n$, then again by
				Lemma~\ref{lem: rein-forest is an expander} we have
				$|E(L\cap S,L\cap\bar S)|\geq
				h_1|(L_3\cup L_4)\cap S|$.
				We treat $(L_1\cap L_2)\cap S$ separately, considering three cases by the sizes of $L_2\cap S$ and $L_1\cap S$.
				\begin{itemize}
					\item
					If $|L_2\cap S|\geq 0.3n$,
					note that each node in $L_2$ is connected by an edge to exactly one node in $L_3$, and these nodes are distinct.
					So there are at least $0.3n$ edges in $(L_2\cap S)\times L_3$, of which at most $0.2n$ have their $L_3$ endpoint in $|L_3\cap S|$, hence $|E(L\cap S,L\cap\bar S)|\geq0.1n\geq 0.05|(L_1\cup L_2)\cap S|$.

					\item
					If $|L_2\cap S|<0.3n$ and $|L_1\cap S|<0.7n$,
					then by Lemma~\ref{lem: rein-forest is an expander} we have
					$|E(L\cap S,L\cap\bar S)|\geq
					h_1|(L_1\cup L_2)\cap S|$.

					\item
					If $|L_2\cap S|<0.3n$ and $|L_1\cap S|\geq0.7n$,
					then by Observation~\ref{obs:strein_moreint},
					$|E(L\cap S,L\cap\bar S)|\geq
					|L_1\cap S|-|L_2\cap S|>0.4n\geq 0.2|L_1\cup L_2|
					\geq 0.4|(L_1\cup L_2)\cap S|$.
				\end{itemize}
			\end{enumerate}
		\end{enumerate}
	\end{enumerate}
\end{enumerate}
\end{proof}

\subsection{Power-law Graph}
\label{ssec:st_powerlaw}

As in the case of maximum matching, we first make our reduction graph robust to degree changes. The following reduction is close in spirit to the expander graph reduction. The graph is as follows:

\begin{itemize}
	\item A $(\log n)$-depth binary forest with a single tree, labelled $L_1 \cup L_2$ as before, with $s$ as the root of the tree.

	\item A $(\log n)$-depth binary forest with $3n$ trees, labelled $L_3 \cup L_4$. We split this binary forest into sets of $n$ trees each, and label them $LU, LM, LL$ (upper, middle, and lower). The roots of each of the $n$ trees are marked as $LX[i]$, for $1 \le i \le n$, $X \in \{U,M,L\}$, and the leaves of the tree with root $LX[i]$ are marked $LX[i,j]$, for $1 \le j \le n$.

	\item A copy of the above structure, with node sets marked $R_i$ instead of $L_i$, respectively. The root of the single tree of $R_1$ is the target node $t$.

    \item For the matrix $M$, add the edges $(LM[i,j],RM[j,i])$ and $(LL[i,j], RL[j,i])$ if $M_{ij}=1$, and add the edges $(LM[i,j], RL[j,i])$ and $(LL[i,j], RM[j,i])$ otherwise.

    \item Given an input vector $u$, for each $i\in[n]$, add the edge $(L_2[i], LM[i])$ if $u_i=1$, and $(L_2[i], LU[i])$ otherwise.

    \item Given an input vector $v$, for each $j\in[n]$, add the edge $(R_2[j], RM[j])$ if $v_j=1$, and $(R_2[j], RU[j])$ otherwise.
\end{itemize}

Note that this fixes the degree distribution present in the graph regardless of the input $(u, M, v)$, unlike in the case of maximum matching where we needed to compensate for the degrees elsewhere in the graph. \Cref{tab:st_powerlaw_degrees} shows the degree distribution on the left side of the reduction graph, regardless of the input bits.

\begin{table}[t]
    \renewcommand{\arraystretch}{1.2}
    \centering
    \begin{tabular}{|c|c|c|c|c|}
    \hline
    Layer & Deg $1$ & Deg $2$ & Deg $3$ \\
    \hline
    $L_1$ & $0$ & $1$ & $n-2$ \\
    \hline
    $L_2$ & $0$ & $n$ & $0$ \\
    \hline
    $L_3$ & $0$ & $2n$ & $3n^2 - 5n$ \\
    \hline
    $L_4$ & $n^2$ & $2n^2$ & $0$ \\
    \hline
    \end{tabular}
    \caption{Degree distribution of the nodes on the left side of the $(s,t)$-distance reduction graph.}
    \label{tab:st_powerlaw_degrees}
\end{table}

Thus the degree distribution in the entire reduction graph is as follows: $(1, 2n^2), (2, 4n^2 + 6n + 2), (3, 6n^2 - 8n - 4)$. As earlier, choose $N$ such that \[
N > \zeta(\beta) \cdot \max \{ (2N_1 + 2n), (2N_2+2n) \cdot 2^{\beta}, (2N_3+2n) \cdot 3^{\beta} \}
,\] and pick any power-law graph $G$ on $N$ nodes. We reduce the count of degree $1, 2$ and $3$ nodes in the graph as earlier, and embed our reduction graph using these nodes. Unlike in the case of matchings, note that the $(s,t)$-distance is not affected by the rest of the graph, and thus $uMv = 1$ if and only if $\dist(s,t) \le 4\log n + 3$.

\section{Lower Bounds for Dynamic Densest Subgraph}
\label{sec:full_Densest_Subgraph_LBs}

For a constant $d\geq 3$, we work with $(2d)$-regular graphs of two different sizes:
$N$-node gadget for each vector entry, and $N^2$-node gadget for each matrix entry.
\begin{definition}[Vector and Matrix Gadgets]
A \emph{vector gadget} is a $6$-edge connected  $2d$-regular graph on $n$ nodes.
A \emph{matrix gadget} is a $6$-edge connected $2d$-regular graph on $n^2$ nodes with one edge removed
\end{definition}
A $6$-edge connected graph is a graph that does not disconnect after the removal of any $5$ edges.
Such a $(2d)$-regular graph is, e.g., a $d$-dimensional torus, or a union of $d$ edge-disjoint cycles, each going through all the nodes.
Note that the density of a vector gadget is $d$, and the density of a matrix gadget is $d - \frac{1}{n^2}$.
We prove the following theorems for maintaining densest subgraphs.
\begin{restatable}{theorem}{denseConst}
	\label{thm:denseConst}
	For any constant $\epsilon>0$,
	there is no dynamic algorithm maintaining an exact densest subgraph,
	on all $N$-node graphs with maximum degree $\Delta \le 7$,
	with amortized \ub{1/4} update time and \ub{1/2} query time,
	unless the \omv{} conjecture is false.
\end{restatable}

\begin{restatable}{theorem}{denseExpand}
	\label{thm:denseExpand}
	For any constant $\epsilon>0$,
	there is no dynamic algorithm maintaining an exact densest subgraph,
	on all $N$-node graphs with constant degree and constant expansion,
	with amortized \ub{1/4} update time and \ub{1/2} query time,
	unless the \omv{} conjecture is false.
\end{restatable}

\subsection{Constant-Degree Graph}
\label{ssec:dense_const_degree}

\subsubsection{The static graph}
\label{sssec:dense_static}

The reduction graph is composed of $2n$ vector gadgets, and $n^2$ matrix gadgets as follows.
\begin{itemize}
	\item $2n$ vector gadgets labelled $U_i$ resp.~$V_i$, for $1 \le i \le n$. The nodes in each gadget are labelled $U_i[j]$ resp.~$V_i[j]$, for $1 \le j \le n$.

    \item $n^2$ matrix gadgets labelled $M_{ij}$, for $1 \le i, j \le n$. The missing edge in the gadget $M_{ij}$ is between the two nodes labelled $M_{ij}[0]$ and $M_{ij}[1]$.

    \item An edge from $U_i[j]$ to $M_{ij}[0]$, and an edge from $V_j[i]$ to $M_{ij}[1]$ for all $1 \le i,j \le n$. Note that this implies that every node of $M_{ij}$ has degree $2d$ and that the total number of edges incident to at least one node of $M_{ij}$ is $d n^2 + 1$.
\end{itemize}
The total number of nodes in the reduction graph is $N = n^4+2n^2 = \Theta(n^4)$.

\subsubsection{Input-dependent edges}
\label{sssec:dense_input}
The above graph is adapted to the specific input instance as follows.
\begin{itemize}
    \item For the matrix $M$, remove one arbitrary edge from the matrix gadget $M_{ij}$ if $M_{ij}=0$.
    \item Given an input vector $u$, for each $i\in[n]$, remove two arbitrary edges from the vector gadget $U_i$ if $u_i=0$.
    \item Given an input vector $v$, for each $j\in[n]$, remove two arbitrary edges from the vector gadget $V_j$ if $v_j=0$.
\end{itemize}

\subsubsection{Densest subgraphs in the graph}
\label{sssec:dense_density}

We use the following simple lemma in our density proof.
\begin{restatable}{lemma}{ratioed}\cite{HenzingerKNS15}
\label{lem:ratioed}
For all numbers $a, b, c, d,$ and $r$, we have:
\begin{enumerate}
    \item If $\frac{a}{b} \ge r$ and $\frac{c}{d} \ge r$, then $\frac{a+c}{b+d} \ge r$.
    \item If $\frac{a}{b} \ge r$ and $\frac{c}{d} \le r$, then $\frac{a-c}{b-d} \ge r$.
\end{enumerate}
\end{restatable}

Let $\rho$ be the density of the densest subgraph in the current graph.
\begin{restatable}{lemma}{densityLem}
\label{lem:densityLem}
If $uMv = 1$, then $\rho \ge d + \frac{1}{n^2+2n}$, and otherwise $\rho < d + \frac{1}{n^2+2n}$.
\end{restatable}
\begin{proof}
($\implies$) First assume that $uMv = 1$. Then there are indices $i, j$ such that $u_i = M_{ij} = v_j = 1$. Consider the subgraph $S = U_i \cup M_{ij} \cup V_j$. It consists of $n^2 + 2n$ nodes and $d(n^2+2n) + 1$ edges. Thus $\rho(S) = d + \frac{1}{n^2+2n}$.

($\impliedby$) Now assume that $uMv = 0$, and that there exists some subset $S \subset V$ with $\rho(S) \ge d + \frac{1}{n^2+2n}$. We first claim that we can modify $S$ in a particular manner without loss of generality, and then derive a contradiction. Specifically, first we  remove from $S$ all subgraphs $M_{ij}$ that are not completely contained in $S$.

\begin{claim}
Let $T$ be a subgraph of a matrix gadget $M_{ij}$ that is not the whole gadget.
Then, after removing $T$ from $S$, it still holds that
$\rho(S) \ge d + \frac{1}{n^2+2n}$.
\end{claim}
\begin{claimproof}
Recall that every node in $M_{ij}$ has degree at most $2d$. Let $q = |T| < n^2$. It follows that there are at most $q d$ edges incident to a node
of $T$ in $G[S]$: If neither $M_{ij}[0]$ nor $M_{ij}[1]$ belong to $T$ then there are most $q d -1$ edges incident to at least on node
of $T$ in $G[S]$ (all being edges between nodes of $T$), as at least one
node of $T$ must have an edge to a node to the rest of $M_{ij}$ that does not belong to $S$. If either $M_{ij}[0]$ or $M_{ij}[1]$, but not both belong to $T$ then there are at most $q d-1$ edges incident to two nodes of $T$ and there is one edge between $T$ and the rest of $S$.
If both $M_{ij}[0]$ and $M_{ij}[1]$ belong to $T$ then there are at most $q d -2$ edges incident to two nodes of $T$ (as there is no edge between
$M_{ij}[0]$ and $M_{ij}[1]$) and there are two edges between $T$ and the rest of $S$.

Thus the removal of $T$ from $S$ removes at most $qd$ edges and $q$ nodes and, hence by Part~2 of~\Cref{lem:ratioed}, $S\setminus T$ has density $\ge d + \frac{1}{n^2+2n}$.
\end{claimproof}

Next we remove all subgraphs $M_{ij}$ such that $M_{ij} = 0$.

\begin{claim}
For all $i, j$, if the bit $M_{ij} = 0$ and $M_{ij}$ is fully contained in $S$, then after removing the corresponding subgraph
from $S$, we still have
$\rho(S)\geq d + \frac{1}{n^2+2n}$.
\end{claim}
\begin{claimproof}
As $M_{ij} = 0$, there are $d n^2$ edges incident to at least one node of $M_{ij}$.
Thus,
 removing $M_{ij}$ from $S$ removes $n^2$ nodes and at most $dn^2$ edges from $G[S]$,
i.e.,~a subgraph of density at most $d$. By Part 2 of~\Cref{lem:ratioed}, $S\setminus M_{ij}$ has density $\ge d + \frac{1}{n^2+2n}$.
\end{claimproof}

Finally we add to $S$ every vector subgraph that is partially contained in $S$.
\begin{claim}
For all $i, j$, after we add any partially contained subgraph $U_i$ or $V_j$
to  $S$ it still holds that
$\rho(S) \ge d + \frac{1}{n^2+2n}$.
\end{claim}
\begin{claimproof}
We only prove our claim for the gadget $U_i$, since the proof for the gadget $V_j$ is analogous. Suppose some subset $U \subset U_i$ is contained in $S$ with $|U| = q<n$.
We will show that adding $A := U_i \setminus U$ to $S$ adds at least $d(n-q)+1$ edges. As $\frac{d(n-q) + 1}{n-q} >  d + \frac{1}{n^2+2n}$, the claim follows from Part 1 of~\Cref{lem:ratioed}.

We are left with showing that adding $A$ to $S$ adds at least $d(n-q)+1$ edges.
Recall that $U_i$ is a $2d$-regular 6-edge connected
graph with $n$ nodes from which 2 edges where removed if $u_i = 0$ and no edges where removed if $u_i = 1$.
We consider the following cases:

\emph{Case 1:} $u_i = 1$ or no removed edge is incident to a node in $A$. In this case every node in $A$ has degree  $2d$ in $G[U_i]$.
Let $a$ be the number of edges between $U$ and $A$, and note that $a\geq 6$ since $U_i$ is 6-edge connected.
Thus, the sum of the degrees of the nodes in $A$ is $G[A] = 2d(n-q) - a$. It follows that there are $d(n-q) - a/2$ edges between  nodes of $A$.
Thus, the total number of edges incident to nodes of $A$ in $G[U_i]$ is $d(n-q) + a/2 \ge d(n-q)+3$ and this is a lower bound of the number of edges that are added to $S$ when $A$ is added.

\emph{Case 2:} $u_i = 0$ and $b$ removed edges belong to
$(A,U)$ and $c$ removed edges have both endpoints in $A$ with $b + c \ge 1$. As only 2 edges are removed, $b + c \le 2$. As $U_i$ is 6-edge connected, let $a \ge 6 - b$ be the number of edges between $U$ and $A$.
In this case the sum of the degrees of the nodes of $A$ in $G[A]$
is $2d(n-q) - (a + b) - 2c$. Thus, the total number of
edges incident to nodes of $A$ in $G[A]$ is $d(n-q) - (a+b)/2 - c$. Hence, the total number of
edges incident to nodes of $A$ in $G[U_i]$ is $d(n-q) - (a+b)/2 - c + a = d (n-q) + (a-b)/2 - c.$ As $a \ge 6 - b$
this is at least $d (n-q) + 3 - b -c$. Since $b + c \le 2$,
it follows that at least $d(n-q) +1$ edges are added to $S$ when $A$ is added.
\end{claimproof}

Thus $S$ has the following structure: It has some matrix gadgets of set bits $M_{ij}$ with all the nodes and both outgoing edges present, and it has some vector gadgets of set or unset bits with all nodes present as well. Let $x$ denote the number of matrix gadgets of set bits contained in $S$, $y$ denote the number of vector gadgets of set bits in $S$, and $z$ denote the number of vector gadgets of unset bits in $S$.
We now consider two cases:

\emph{Case 1:} $y + z = 0$. Then $S$ consists only of $x$ matrix gadgets and no vector gadgets. Thus $\rho(S) = \frac{x \cdot (d n^2 -1)}{n^2x} < d$, which is a contradiction

\emph{Case 2:}  $y + z > 0$.
Then the density of $S$ is given by \[
\rho(S) = \frac{x\cdot (dn^2+1) + y\cdot dn + z \cdot (dn - 2)}{n^2x+ny+nz}
\] We claimed that $\rho(S) \ge d + \frac{1}{n^2+2n}$, which is the same \[
2nx \ge ny + (2n^2+5n) \cdot z
\] Since $uMv = 0$, for every matrix gadget $M_{ij}$ of a set bit, either the corresponding $U_i$ or $V_j$ must be unset. Thus we assign at most $n$ matrix gadgets to each unset vector gadget, giving us that $x \le nz$, giving \[
2n^2 z \ge 2nx \ge ny + (2n^2 + 5n) \cdot z
,\] which is a contradiction as desired since $y + z > 0$.
\end{proof}

\subsubsection{Complexity of the reduction}
\label{sssec:dense_complexity}
We are now ready to prove the theorem.

\denseConst*
\begin{proof}
Consider the reduction graph above with $d=3$. It consists of $N = n^4 + 2n^2 = \Theta(n^4)$ nodes. Every time we get a new $(u,v)$ input vector pair, we reinsert all the removed edges in the vector gadgets $U_i, V_i$, for all $1 \le i \le n$, and then delete edges according to the new input vectors. This takes $O(n)$ updates in total. After that, we query once for the density of the densest subgraph in this new graph, and return $1$ if and only if $\rho \ge 3 + \frac{1}{n^2+2n}$.

Thus for each pair of input vectors, we perform $O(n)$ updates and $O(1)$ query. In total, checking $n$ pairs takes us $O(n^2)$ updates and $O(n)$ query. If there were an algorithm for maintaining the density of the densest subgraph on constant-degree graphs with update time \ub{1/4} (i.e., $O(n^{1-4\epsilon})$) and query time \ub{1/2} (i.e., $O(n^{2-4\epsilon})$), then we can decide if $uMv=1$ for all $n$ pairs in $O(n^{3-4\epsilon})$ time, contradicting the \omv{} conjecture.
\end{proof}

\subsection{Expander Graph}
\label{ssec:dense_expand}

The extension required to make the reduction hold even for expanders for the densest subgraph problem is simpler than the corresponding extensions for the previous problems.
Pick $d$ such that there is a $d'$-regular expander for some $d'\leq d-2$.
\label{sssec:dense_expand_static}
The static graph is constructed as follows.
\begin{itemize}
	\item The reduction graph $G_0$ from \Cref{sssec:dense_static} as a subgraph together with the input-dependent edges from \Cref{sssec:dense_input}.

	\item A $(d')$-regular expander graph $G_1$ on $n^4+2n^2$ nodes,
	for some $d'\leq d-2$,
	with constant expansion $h_0$, together with a bijection $\pi$ from the nodes of $G_0$ to the $G_1$.

    \item An edge between every node $v$ of $G_0$ and the node $\pi(v)$ of $G_1$, giving a perfect matching  connecting the nodes of $G_0$ with these of $G_1$.
\end{itemize}

The total number of nodes in the reduction graph is $N = 2n^4+4n^2 = \Theta(n^4)$.

The density arguments in this reduction graph are similar to the one in the constant-degree case, as in \Cref{lem:densityLem}. The proof follows from the following fact: for any subset $T$ with density $> d$, removing all of the nodes of $G_1$ from $T$ cannot decrease the density of $T$, since each node in $G_1$ has degree at most $d-1$ in $G[T]$. Thus, in the proof of \Cref{lem:densityLem} in the setting where $uMv = 0$,  we add a first step removing all the nodes of $G_1$ from the subset $S$ and then proceed as before. This leads to the same lower bounds as in the constant-degree case. We have the following theorem for expander graphs:

\denseExpand*

To verify that the graph is indeed an $h_1$-expander graph for some constant $h_1>0$, we follow a similar pattern as in the proof of expansion for maximum matching in~\Cref{ssec:match_expand_proof}.
Consider a reduction graph with $ V =U \cup U'$, where $U'$ is the node set of the expander $G_1$. Let $S$ be some set of nodes (we do not bound $|S|$). We have the following observations.
\begin{observation}
	\label{obs:drein_medium}
	Let $0<c<1/2$, that may depend on $n$.
	If $c|U'|\leq|S\cap U'|\leq (1-c)|U'|$, then
	$E(S,\bar S)|\geq ch_0|U'|$.
\end{observation}
\begin{claimproof}
	Consider first the case $|S\cap U'|\leq |U'|/2$, in which the expander graph edges on $U'$ guarantee $|E(S,\bar S)\cap U'|\geq h_0|S|$, implying
	$|E(S,\bar S)|\geq|E(S,\bar S)\cap U'|\geq h_0|S|\geq ch_0|U'|$.
	Otherwise, apply the analogous argument on the set $U'\setminus S$.
\end{claimproof}

\begin{observation}
	\label{obs:drein_moreint}
	$|E(S,\bar S)|\geq |U\cap S|-|U'\cap S|$.
\end{observation}
\begin{claimproof}
	Every node in $U\cap S$ has a matching partner (by the bijection used to map the nodes of $U$ to $U'$) in $U'$.
	In total, each of these nodes has $|U\cap S|$ edges going to their matching partners, of which at most $|U'\cap S|$ partners are in $S$.
	Hence, $|E(S,\bar S)|\geq |U\cap S|-|U'\cap S|$.
\end{claimproof}

\begin{observation}
	\label{obs:drein_lessint}
	Let $0<c<1/2$ constant.
	If $|U\cap S|\leq c|U'\cap S|$
	then $|E(S,\bar S)|\geq (1-c)|U'\cap S|$.
\end{observation}
\begin{claimproof}
	If $|U\cap S|\leq c|U'\cap S|$, consider the matching partners of the nodes in $U'\cap S$: there are at least $|U'\cap S|$ such partners in $U$, of which at most $c|U\cap S|$ are in $S$, and the others are in
	$U\cap \bar S$. Hence
	$|E(U\cap \bar S,U'\cap S)|\geq (1-c)|U'\cap S|$.
\end{claimproof}

\begin{lemma}
The reduction graph is an expander.
\end{lemma}
\begin{proof}
Now fix a set $S \subseteq V$ of nodes, with $|S|\leq n^4 + 2n^2 = (|U\cup U'|)/2$.
We consider different ranges of $|U'\cap S|$, and show that for each of them, $|E(S,\bar S)|\geq h_1|S|$ for some constant $h_1$.
\begin{enumerate}
    \item
    If $|U'\cap S|>0.9|U'|$, the inequalities $|U|\le |U'|$ and $|S|\leq|U\cup U'|/2$ imply $|U\cap S|\le 0.1|U'|$, and hence
    $|U'\cap S|>(1/9)|U\cap S|>0.1|U\cap S|$.
    Observation~\ref{obs:drein_lessint} implies $|E(S,\bar S)|\geq 0.9|U'\cap S|\geq 0.81|U'|> 0.4|U\cup U'|\geq 0.8|S|$.

    \item
    If $0.1|U'|\leq |U'\cap S|\leq 0.9|U'|$, then by Observation~\ref{obs:drein_medium} we have
    $|E(S,\bar S)|\geq 0.1h_0|U'|$.
    Since $|U|\le |U'|$ and $|S|\leq|U\cup U'|/2$, the above implies
    $|E(S,\bar S)|\geq 0.05h_0|U\cup U'|\geq 0.1 h_0|S|$.

    \item
    If $|U'\cap S|<0.1|U'|$, then
    \begin{enumerate}
        \item
        If $|U\cap S|>2|U'\cap S|$,
        add
        $0.5|U\cap S|-1.5|U'\cap S|$ to both sides of the inequality and multiply by $2/3$ to get
        $|U\cap S|-|U'\cap S|>1/3(|U\cap S|+|U'\cap S|)=|S|/3$.
        By Observation~\ref{obs:drein_moreint},
        $|E(S,\bar S)|\geq |U\cap S|-|U'\cap S|>|S|/3$.

        \item
        If $|U\cap S|\leq 2|U'\cap S|$
        then
        $|S|=|U\cap S|+|U'\cap S|\leq 3|U'\cap S|$.
        The expander edges on $U'$ guarantee that
        $|E(S,\bar S)|\geq h_0|U'\cap S|\geq(h_0/3)|S|$.
        \qedhere
    \end{enumerate}
\end{enumerate}
\end{proof}

\subsection{Power-law Graph}
\label{ssec:dense_powerlaw}

Let $d = 3$. We show our densest subgraph lower bounds for $\beta > 2.74$. We use a reduction graph similar to the one in the constant-degree case for our reduction. Our reduction graph consists of the vector and matrix gadgets as before, along with new nodes we introduce to moderate the degree. We add a cycle $C_i^u$ on four new nodes for each vector gadget, and a cycle $C_{ij}$ on four new nodes for each matrix gadget.

\begin{itemize}
	\item $2n$ vector gadgets labelled $U_i$ resp.~$V_i$, for $1 \le i \le n$. The nodes in each gadget are labelled $U_i[j]$ resp.~$V_i[j]$, for $1 \le j \le n$.

    \item $2n$ cycles $C_i^u$ and $C_j^v$ on $4$ nodes each, with the nodes labelled $C_i^w[x]$ for $w \in \{u, v\}$, $x \in \{a, b, c, d\}$

    \item $n^2$ matrix gadgets labelled $M_{ij}$, for $1 \le i, j \le n$. The missing edge in the gadget $M_{ij}$ is between the two nodes labelled $M_{ij}[0]$ and $M_{ij}[1]$.

    \item $n^2$ cycles $C_{ij}$ on four nodes $C_{ij}[x]$ for $x \in \{a,b,c,d\}$ each.

    \item An edge from $U_i[j]$ to $M_{ij}[0]$, and an edge from $V_j[i]$ to $M_{ij}[1]$ for all $1 \le i,j \le n$.

    \item If $M_{ij} = 0$, remove two arbitrary edges, say, $(M_{ij}[a], M_{ij}[b]), (M_{ij}[c], M_{ij}[d])$, and the edges $(C_{ij}[a], C_{ij}[b]), (C_{ij}[c], C_{ij}[d])$, and add the edges $(M_{ij}[x], C_{ij}[x])$ for $x \in \{a, b, c, d\}$ to the graph.

    \item Given an input vector $u$, for each $i\in[n]$, do the following if $u_i = 0$. Remove two arbitrary edges from the vector gadget $U_i$, say $(U_i[a], U_i[b]), (U_i[c], U_i[d])$, and also remove the edges $(C_i^u[a], C_i^u[b]), (C_i^u[c], C_i^u[d])$. Add the edges $(U_i[x], C_i^u[x])$ for $x \in \{ a, b, c, d\}$. Similarly for an input vector $v$.
\end{itemize}

Note that each node in a vector gadget always has degree $2d+1$, and each node in a matrix gadget has degree $2d$, regardless of input. Further, the degrees of the nodes in all the cycles are exactly $2$ for all inputs. The degree distribution of the reduction graph is $(2, 4n^2+8n), (2d, n^4), (2d+1, 2n^2) $. Choose $N$ such that \[
N > \zeta(\beta) \cdot \max \{ (4n^2 + 8n) \cdot 2^{\beta}, (n^4) \cdot (2d)^{\beta}, (2n^2) \cdot (2d+1)^{\beta} \}
\]
In a power-law graph with $N$ nodes, the sum of degrees of all nodes is given by $\zeta(\beta - 1) \cdot N / \zeta(\beta)$, while the number of nodes of degree $1$ is $N/\zeta(\beta)$. Thus for all $\beta$ such that $\zeta(\beta-1) < 2$, more than half the total degree comes from nodes of degree $1$, since
\begin{equation}
\label{eqn:dense_pl}
N'_1 = \frac{N}{\zeta(\beta)} > \frac{\zeta(\beta-1)}{2} \cdot \frac{N}{\zeta(\beta)} = \frac{\sum_{v \in V} \deg(v)}{2}
\end{equation}
This is true in particular for all $\beta > 2.74$. Thus now we first construct the reduction graph as above on $n^4 + 6n^2 + 8n$ nodes, and this does not exceed $N'_d$ for any $d$ because of the definition of $N$. Now for every other degree that needs to be satisfied of degree $d'> 1$, we simply add a new node and create a star with $d'$ nodes of degree $1$ attached to it. Note that we can do this for all the remaining nodes because of \Cref{eqn:dense_pl}. If there are any remaining degree $1$ requirements to be satisfied, we simply add a perfect matching on that many nodes.

None of the stars or the perfect matchings can be part of a subgraph of density $>d$, since they all have density $< 1$. Thus any subgraph of density $> d$ must be from our reduction graph. Further, note that the cycle gadgets can be removed from any subgraph of density $> d$, since they have degree exactly $2$. Thus $\rho \ge d + \frac{1}{n^2+2n}$ if and only if $uMv = 1$, which proves our claim.

\bibliographystyle{plainurl}
\bibliography{References}

\appendix

\newpage

\section{Related Work}
\label{sec:Related_Work}
We describe further related work in this section.

\emph{Planar graphs.}
Prior work  for fully dynamic shortest paths in planar graphs is summarized in
 Table~\ref{tab:planar_lbs} and Table~\ref{tab:planar_ubs}. There is one further lower bound result in dynamic planar graphs, namely for bipartite maximum weigthed matching, showing a tradeoff of
 $\max \{u, q \} = \Omega(N^{1/2-\epsilon})$, where $u$ denotes the update time and $q$ the query time~\cite{AbboudD16}. The  planar graphs used in these lower bound constructions all have constant degree.

 There exists also further work on upper bounds in planar graphs.
 Italiano et al.~~\cite{DBLP:conf/stoc/ItalianoNSW11} designed a fully dynamic algorithm for maximum flow and minimum cut with $\tilde O(N^{2/3})$ update time in planar graphs. In the deletions-only setting in directed graphs Italiano et al.~\cite{DBLP:conf/stoc/ItalianoKLS17} gave an algorithm with $\tilde O(1)$ time per operation.

 \emph{Other graph classes.}
 In $\sqrt{N}$-separable
graphs Goranci et al.~\cite{DBLP:conf/esa/GoranciHP18} give almost tight upper and lower bounds for maintaining $(1+\epsilon)$-approximations of the all-pairs effective resistances.

In graphs with constant arboricity Peleg and Salomon~\cite{Peleg2016} gave $(1+\epsilon)$-approximate matching algorithm in constant time.

\emph{Insertions-only and deletions-only lower bounds.}
In general graphs Dahlgaard~\cite{dahlgaardhardness2016} presented lower bounds of $\Omega(N^{1-o(1)})$ for incremental or decremental maximum cardinality bipartite matching,
of $\Omega(m^{1-o(1)})$
for incremental or decremental maximum flow in directed and weighted sparse graphs,
and $\Omega(N^{1/2-o(1)})$
for incremental or decremental $(4/3 - \delta)$-approximating the diameter of an unweighted graph for any small constant $\delta > 0$.
These lower bounds for diameter were later improved
in~\cite{DBLP:conf/icalp/AnconaHRWW19}. Results for dynamic near-additive spanners were given
in~\cite{DBLP:conf/soda/BergamaschiHGWW21}.

\emph{Sensitivity model.}
Ancona~\cite{Ancona2019} studied diameter approximation and related problems in dynamic constant-degree graphs.
She focused on a different model than ours, called the sensitivity model, and similarly to us, proved conditional lower bounds using reductions to the \omv{} conjecture.

\begin{table}
    \renewcommand{\arraystretch}{1.2}
    \centering
    \begin{tabular}{|c|c|c|c|}
    \hline
    \textbf{Problem} & \textbf{Ref.} & \textbf{Assuming} & \textbf{LBs} \\ \hline
    APSP weighted & \cite{AbboudD16} & APSP conjecture & $u \cdot q = \Omega(N^{1-o(1)})$ \\ \hline
    APSP unit weight & \cite{AbboudD16} & \omv{} conjecture & $\max \{u^2 \cdot q, q^2 \cdot u \} = \Omega(N^{1-o(1)})$ \\ \hline
    $(s,t)$-distance, girth, diameter & \cite{AbboudD16} & \omv{} conjecture & $\max \{u, q \} = \Omega(N^{1/2-\epsilon})$ \\ \hline
    \end{tabular}
    \caption{Lower bounds for fully dynamic shortest-paths algorithms in planar graphs, where $u$ denotes the time per update and $q$ the time per query.}
    \label{tab:planar_lbs}
\end{table}

\begin{table}
    \renewcommand{\arraystretch}{1.2}
    \centering
    \begin{tabular}{|c|c|c|c|}
    \hline
    \textbf{Problem} & \textbf{Ref.} & \textbf{Update} & \textbf{Query} \\ \hline
    undirected $1+\epsilon$-approx. $(s,t)$-distance  & \cite{klein1998fully} & $\tilde O (n^{2/3})$ & $\tilde O (n^{2/3})$ \\ \hline
    undirected $1+\epsilon$-approx. $(s,t)$-distance & \cite{abraham12fully} & $\tilde O (n^{1/2})$ & $\tilde O (n^{1/2})$ \\ \hline
    undirected $(s,t)$-distance with treewidth $k$ & \cite{Abraham2016} & $O (k^3\log n)$ & $O (k^2\log n \log (k \log n))$ \\ \hline
    SSSP on weighted digraphs & \cite{CharalampopoulosK20} & $\tilde O (n^{4/5})$ & $O(\log^2 n)$ \\ \hline
    \end{tabular}
    \caption{Upper bounds for fully dynamic shortest-path algorithms in planar graphs}
    \label{tab:planar_ubs}
\end{table}

\section{Partially Dynamic Lower Bounds}
\label{sec:partial_lbs}

\subsection{\texorpdfstring{Dynamic $(s,t)$-Distance}{Dynamic (s,t)-Distance}}
\label{ssec:st_partial}

We show that our dynamic $(s,t)$-distance lower bound for constant-degree graphs also holds for partially dynamic algorithms using the following reduction graph.

\begin{itemize}
	\item A $(\log n)$-depth binary forest with a single tree. The internal nodes are marked $L_1$ and the leaves $L_2$; the root of the tree is the source node $s$, and the nodes of $L_2$ are marked as $L_2[i]$ for $1\le i \le n$.

	\item A $(\log n)$-depth binary forest with $n$ trees. The internal nodes are marked $L_3$ and the leaves $L_4$. The roots of each of the $n$ trees are marked as $L_3[i]$, for $1 \le i \le n$, and the leaves of the tree with root $L_3[i]$ are marked $L_4[i,j]$, for $1 \le j \le n$.

    \item An edge from $L_2[i]$ to $L_3[i]$ for $1 \le i \le n$.

    \item $n$ paths $P[i]$ on $n$ nodes. The nodes of the path $P[i]$ are marked $P[i,n+1-j]$, $1 \le j \le n$. Edges from $L_4[i,j]$ to $P[i,j]$ for all $1 \le i,j \le n$.

	\item A $(\log n)$-depth binary forest with $n$ trees. The internal nodes are marked $L_5$ and the leaves $L_6$. The roots of each of the $n$ trees are marked as $L_5[i]$, for $1 \le i \le n$, and the leaves of the tree with root $L_5[i]$ are marked $L_6[i,j]$, for $1 \le j \le n$.

    \item An edge from $P[i,1]$ to $L_5[i]$ for all $1 \le i \le n$.

	\item A copy of the above structure, with node sets marked $R_i$ and $Q[i]$ instead of $L_i$ and $P[i]$. The root of the single tree of $R_1$ is the target node $t$.

    \item For the matrix $M$, add the edge $(L_6[i,j],R_6[j,i])$ if $M_{ij}=1$.

	\item Deletion of edges for each input pair $(u^j, v^j)$ as detailed next.
\end{itemize}

We perform the following deletions upon the arrival of the $j^{th}$ input vectors $(u^j, v^j)$
\begin{itemize}
    \item Given the $j^{th}$ input vector $u^j$, for each $i\in[n]$, delete the edge $(L_4[i,j], P[i,j])$ if $u_i^j=0$.
    \item Given the $j^{th}$ input vector $v^j$, for each $i\in[n]$, delete the edge $(R_4[i,j], Q[i,j])$ if $v_i^j=0$.
\end{itemize}

Before the $(j+1)^{th}$ input vector arrives, delete all $(L_4[i,j], P[i,j])$ and $(R_4[i,j], Q[i,j])$ edges for all $1 \le i \le n$. It is easy to see that there is a path of length $6\log n + 5 + 2j$ if and only if $u^jMv^j = 1$.

Since the reduction graph consists of $\Theta(n^2)$ nodes and we make $O(n)$ updates and $1$ query for each input pair, we get the same lower bounds as in the fully-dynamic setting. Note that this lower bound can be made to work for the insertions only setting as well by reversing the path $P$ and $Q$.

\subsection{Dynamic Maximum Matching}
\label{ssec:match_partial}

We use the following reduction graph to make our fully dynamic maximum matching lower bounds hold for the partially dynamic setting as well.

\begin{itemize}
	\item A reduction gadget with $n$ subgadgets of size $2n+2$, on a set $L_1 \cup L_2$. The subgadgets are labelled $LE[j]$ for $1 \le j \le n$, and the nodes of subgadget $LE[j]$ are labelled $L_1[j,i]$ or $L_2[j,i]$ for $0 \le i \le n$. The path in each subgadget goes from $L_1[j,0]$ to $L_2[j,n]$.

	\item A reduction gadget with $n$ subgadgets of size $2n+2$, on a set $L_3 \cup L_4$. The subgadgets are labelled $LF[i]$ for $1 \le i \le n$, and the nodes of subgadget $LF[i]$ are labelled $L_3[i,j]$ or $L_4[i,j]$ for $0 \le j \le n$. The path in each subgadget goes from $L_3[i,0]$ to $L_3[i,n]$.

    \item Edges from $L_2[j, i]$ to $L_3[i,j]$ for all $1 \le i, j \le n$.

	\item A reduction gadget with $n$ subgadgets of size $2n+2$, on a set $L_5 \cup L_6$. The subgadgets are labelled $LG[i]$ for $1 \le i \le n$, and the nodes of subgadget $LG[i]$ are labelled $L_5[i,j]$ or $L_6[i,j]$ for $0 \le j \le n$. The path in each subgadget goes from $L_5[i,0]$ to $L_6[i,n]$.

    \item An edge from $L_4[i,n]$ to $L_5[i,0]$ for all $1 \le i \le n$.

    \item A copy of the above structure, with node sets marked $R_i$ instead of $L_i$.

    \item For the matrix $M$, add the edge $(L_6[i,j],R_6[j,i])$ if $M_{ij}=1$.

	\item Deletion of edges for each input pair $(u^j, v^j)$ as detailed next.
\end{itemize}

We perform the following deletions upon the arrival of the $j^{th}$ input vectors $(u^j, v^j)$:

\begin{itemize}
    \item Delete the edge $(L_1[j,0], L_2[j,0])$.
    \item Given the $j^{th}$ input vector $u^j$, for each $i\in[n]$, delete the edge $(L_2[j,i], L_3[i,j])$ if $u_i^j=0$.
    \item Given the $j^{th}$ input vector $v^j$, for each $i\in[n]$, delete the edge $(R_2[j,i], R_3[i,j])$ if $v_i^j=0$.
\end{itemize}

Before the $(j+1)^{th}$ input vector arrives, delete the edges $(L_2[j,0], L_1[j,1])$, $(R_2[j,0], R_1[j,1])$, $(L_2[j,i], L_3[i,j])$, and $(R_2[j,i], L_3[i,j])$ for all $1 \le i \le n$. It is easy to see that there is a matching with only $4j-2$ nodes unmatched if and only if $u^j M v^j = 1$.

Since the reduction graph consists of $\Theta(n^2)$ nodes and we make $O(n)$ updates and $1$ query for each input pair, we get the same lower bounds as in the fully-dynamic setting. This reduction can be made to work for the insertions only setting by reversing the above construction.
\end{document}